\documentclass{article}

\usepackage{arxiv}
\usepackage[sort,numbers]{natbib}
\usepackage[utf8]{inputenc} 
\usepackage[T1]{fontenc}    
\usepackage{hyperref}       
\usepackage{url}            
\usepackage{booktabs}       
\usepackage{amsfonts}       
\usepackage{amsmath,amsthm,amssymb}        
\usepackage[cal=cm,scr=euler]{mathalpha}
\usepackage{nicefrac}       
\usepackage{microtype}      
\usepackage{authblk}
\usepackage{graphicx}
\usepackage{subfigure}
\usepackage{caption}
\usepackage{tablefootnote}
\usepackage[linesnumbered,ruled]{algorithm2e}

\newcommand{\N}{{\mathbb{N}}} 
\newcommand{\R}{{\mathbb{R}}} 
\newcommand{\Rd}{{\mathbb{R}^d}} 
\newcommand{\Hb}{{\mathcal{H}}} 
\newcommand{\V}{{\mathcal{V}}} 
\newcommand{\TV}{\prod^{\infty}_{k=0}(\R^d)^{\otimes k}} 
\newcommand{\TVm}{\prod^{m}_{k=0}(\R^d)^{\otimes k}} 
\newcommand{\THm}{\prod^{m}_{k=0}\Hb^{\otimes k}} 
\newcommand{\Xt}{{\mathcal{X}}} 
\newcommand{\M}{{\mathcal{P}(\Xt)}} 
\newcommand{\Pspc}{{\mathscr{P}}} 
\newcommand{\D}{{\mathscr{P}}} 
\newcommand{\Data}{{\mathcal{D}}} 
\newcommand{\Hz}{{\mathcal{H}_p}} 
\newcommand{\Sig}{{\mathbf{S}}} 
\newcommand{\Sigkm}{{k^m_{\text{sig}}}} 
\newcommand{\x}{{\mathbf{x}}} 
\newcommand{\y}{{\mathbf{y}}} 
\newcommand{\z}{{\mathbf{z}}} 
\newcommand{\mmd}{\widehat{\text{MMD}^2}} 

\newtheorem{proposition}{Proposition}

\theoremstyle{definition}
\newtheorem{definition}{Definition}
\newtheorem{remark}{Remark}

\title{Generative modelling of financial time series with structured noise and MMD-based signature learning}
\author[ ]{Lu Chung I\textsuperscript{1}\thanks{Corresponding author: \texttt{lu.chung.i@u.nus.edu}}}
\author[2]{Julian Sester}
\affil[1]{National University of Singapore, 21 Lower Kent Ridge Road, Singapore 119077, \protect\\ Asian Institute of Digital Finance (AIDF)}
\affil[2]{National University of Singapore, 21 Lower Kent Ridge Road, Singapore 119077, \protect\\ Department of Mathematics}

\begin{document}

\maketitle

\begin{abstract}{}
    Generating synthetic financial time series data that accurately reflects real-world market dynamics holds tremendous potential for various applications, including portfolio optimization, risk management, and large scale machine learning. We present an approach that {uses structured noise} for training generative models for financial time series. The expressive power of the signature transform {has been shown to be able} to capture the complex dependencies and temporal structures inherent in financial data {when used to train generative models in the form of a signature kernel }. We employ a moving average model to model the variance of the noise input, enhancing the model's ability to reproduce stylized facts such as volatility clustering. Through empirical experiments on S\&P 500 index data, we demonstrate that our model effectively captures key characteristics of financial time series and outperforms comparable {approaches}. In addition, we explore the application of the synthetic data generated to train a reinforcement learning agent for portfolio management, achieving promising results. Finally, we propose a method to add robustness to the generative model by tweaking the noise input so that the generated sequences can be adjusted to different market environments with minimal data.
\end{abstract}

\section{Introduction} \label{sec:intro}

The availability of high-quality, realistic synthetic data for financial time series holds immense potential for various quantitative finance applications (see e.g. \cite{assefa2020generating}).
However, capturing the complex and nuanced characteristics of financial markets, such as volatility clustering, fat tails, and long-range dependencies, poses a significant challenge for traditional generative models.
Recent advances in deep learning have introduced promising techniques, including Variational Autoencoders (VAEs) introduced in \cite{kingma2013auto} and Generative Adversarial Networks (GANs) pioneered in \cite{goodfellow2014generative}, for generating synthetic financial data.

{The signature of a path has been used effectively as a means to train generative models for financial time series using the maximum mean discrepancy (MMD) as defined in \cite{gretton2012kernel}.
This is done with signature kernels, for example in \cite{issa2023non} by using computational methods proposed in \cite{kiraly2019kernels} and \cite{salvi2021signature}.
The prior works in this area typically use identical and independent Gaussian noise as the key stochastic input to these models.
This paper proposes enhancing the model's ability to represent stylized facts by employing a moving average model to generate structured noise input, leading to more realistic synthetic data generation.}

To demonstrate the effectiveness of our approach, we compare the characteristics of the generated sequence using our MMD approach against {similar approaches using the signature kernel as well as an alternative approach} from \cite{xu2020cot} that uses causal optimal transport (COT), which is a type of Wasserstein distance that respects causality (see \cite{backhoff2022estimating}).
{Signature kernel methods have been applied in \cite{liao2024sig} through the Conditional Sig-WGAN framework, and in \cite{issa2023non}, where neural SDEs were incorporated into the model architecture.}
We find that although all the models are generally able to capture most stylised facts of financial time series, there are some improvements in the statistical aspects of our MMD generated sequence.
Furthermore, we explore the application of the generated synthetic data for training a reinforcement learning agent on the task of portfolio management, achieving promising results.
Finally, we propose a method to add robustness to the generative model by tweaking the noise input so that the generated sequences can be adjusted to different market environments with minimal data.

The sections are organised as follows.
Section \ref{sec:related_works} provides an overview of related works in the field of generative models for financial time series that use neural networks.
Section \ref{sec:preliminaries} details the technical underpinnings behind the methodology.
We describe our proposed generative model and the modelling of the noise variance in Section \ref{sec:generative_model}.
Section \ref{sec:experiments} presents the experimental results, comparing our approach with {the alternatives}, and evaluating the performance of the synthetic data generated for portfolio management using reinforcement learning.
We also introduce a method to add robustness to the generative model in this section. Finally, Section \ref{sec:conclusion} concludes the paper and discusses potential future research directions.

\section{Related Works} \label{sec:related_works}

Generative models are a type of statistical model that can be used to generate new data instances. They are especially useful in the financial domain, where they can be used to generate financial time series data. {Financial time series data is a sequence of values that are either measured at regular intervals e.g. index price series, or in some cases at irregular time intervals e.g. limit order books.} Generative models can be used to generate new time series data that is similar to the original data. This is useful for a variety of applications, including portfolio optimization, risk management, and privacy protection. In addition, generative models can be used to backtest trading strategies or serve as training data for machine learning models. Traditional statistical methods for modelling financial time series include autoregressive models, GARCH models, and stochastic volatility models.
We focus instead on related works that also use neural networks as the backbone for the generative model.

One of the earlier work on financial time series is from \cite{kondratyev2019market} which used restricted Boltzmann machines to generate foreign exchange rates. More recent advances in deep learning include approaches such as the Variational Autoencoders (VAEs) where the generative model is trained to maximise the evidence lower bound (ELBO) of the reference data as introduced in \cite{kingma2013auto}.
In \cite{buehler2020data}, the authors trained VAEs to reproduce the log signatures of the reference data with a diagonal Gaussian distribution for the latent space.
The choice of log signature is motivated by the fact that it is a unique feature transformation that captures the key characteristics of the path.
The challenge lies in inverting the log signature to the original path, which is highly nontrivial (see e.g. \cite{lyons2018inverting}).
In \cite{wiese2021multi} the authors used autoencoders for dimensionality reduction of discrete local volatilities and utilises normalising flows to model the distribution of the latent variables while enforcing no-arbitrage constraints.

Another approach uses diffusion based models (see \cite{ho2020denoising}). In \cite{kong2020diffwave}, this was used to generate audio data in wave form. More recently, in \cite{huang2024generative}, the FTS-Diffusion model was proposed to generate financial time series with irregular and scale-invariant patterns.

An even more popular approach is the framework of Generative Adversarial Networks (GANs) pioneered in \cite{goodfellow2014generative} to train neural networks generators of financial time series.
One of the earliest works using this approach is \cite{takahashi2019modeling} where the authors used multi-layer perceptrons and convolutional networks in the generator and discriminator.
The QuantGAN from \cite{wiese2020quant} focuses on the choice of neural network architecture.
In particular, an architecture called Temporal Convolutional Networks (TCN) was used to model drift and volatility processes in order to generate samples.
\cite{fu2022simulating} further incorporates the attention mechanism with transformers into the network architecture under the GAN framework.
One common issue with the GAN framework is that it is {known to} exhibit instabilities during training.
The authors of QuantGAN report that "training was very irregular and did not converge to a local optimum of the objective function", hence the need for model selection across multiple parameters saved during training.

There are other GAN type approaches such as the TimeGAN from \cite{yoon2019time}, which blends the GAN framework with autoencoders and supervised learning by simultaneously training the autoencoder to create the latent space and the generator and discriminator to model the dynamics of the latent space.
TailGAN from \cite{cont2022tail} focuses on generating time series that preserve the tail risk features of the training data by exploiting the joint elicitability property of value-at-risk and expected shortfall.
In \cite{kidger2021neural}, the authors use neural SDEs (see \cite{kidger2022neural}), with the Wasserstein-1 distance to train the generator similar to the Wasserstein GAN from \cite{arjovsky2017wasserstein}.
This falls under the GAN framework as the Wasserstein-1 distance is calculated using an optimisation based formulation that uses neural networks as function approximators.
The COT-GAN from \cite{xu2020cot}, uses causal Wasserstein or Causal Optimal Transport (COT) (see \cite{backhoff2022estimating}) which adds an additional constraint to respect causality.
Although this particular work was not applied to financial time series, it is an example of a GAN based approach that uses a distance metric as the discriminator exhibiting better stability during training.
More recently, \cite{acciaio2024time} proposed a time-causal VAE (TC-VAE) which uses VAEs and causal Wasserstein distance to train the generator for financial time series.

There are also works that use Wasserstein like distances as the scoring function for training the generator.
In \cite{ni2021sig}, the authors use the universality property of path signatures to replace the space of Lipschitz functions with linear functionals on the signature space in the Kantorovich-Rubinstein dual formulation of the Wasserstein distance.
This switch is critical as it admits an analytical solution to the optimisation problem which means the training of the generator is no longer adversarial and belongs to the class of Score-based Generative Models (SGMs) described in \cite{song2019generative}.
The approach was further extended to the conditional setting with the Conditional Sig-WGAN in \cite{liao2024sig} and in \cite{lozano2023neural}.
In \cite{biagini2024universal} the authors show that randomised signatures (see \cite{compagnoni2023effectiveness}), which is related to the signature transform and has the same expressive power while being finite-dimensional, also have the universality property. Therefore, it can be used to construct the Wasserstein distance as in \cite{ni2021sig} for training the generator.

The closest related work is \cite{issa2023non} where the authors used neural SDEs as the base model and train the generator using a modified version of the MMD with the signature kernel.
Similar to our approach, the authors also use the insight from \cite{kiraly2019kernels} to first lift the path to the canonical feature space of a classical kernel such as the Gaussian kernel.
One key difference to the approach in this paper is the distribution of the noise input to the generative model.
It is typical in generative models to use independently distributed Gaussian noise for each time step.
We use a moving average model to model the variance of the noise which is found to be important in generating autocorrelation in the absolute returns of the generated sequences which is a stylised fact of financial time series.
The use of the moving average model is novel to the best of our knowledge.
{
The additional structure in the noise also allows us to quickly adapt the generative model to produce time series that are better aligned to different market environments by only adjusting the noise distribution.
This adds efficiency by requiring less data to refit the model and is also quick to implement in practice.
The architecture of the neural network is also different to \cite{issa2023non} as we use a recurrent neural network (RNN) based architecture instead of a neural SDE.
Specifically, we use a Long Short-Term Memory (LSTM) network as the backbone of the generator, which provides a structured way of capturing the temporal dependencies in the data.
Finally, we perform direct comparisons between our approach and related works such as the COT-GAN, Conditional Sig-WGAN and the signature kernel trained neural SDEs from \cite{xu2020cot}, \cite{liao2024sig} and \cite{issa2023non} respectively, to demonstrate the effectiveness of our approach.
}

In terms of applying the generative model to downstream tasks, the authors of \cite{coletta2021towards} simulate limit order book data while in \cite{koshiyama2021generative}, the authors create synthetic data for building or {fine tuning} trading strategies.
The strategies are based on data driven machine learning techniques such as neural networks and gradient boosting machines.
We will do a similar portfolio management task using reinforcement learning trained on synthetic data created by the generative model.
In particular, we look to generate synthetic data of longer time horizons which is another aspect that differentiates our approach from the existing works.

\section{Preliminaries} \label{sec:preliminaries}

In the following subsections, we explain the key concepts behind our methodology which are the reproducing kernel Hilbert space (RKHS), maximum mean discrepancy (MMD) and path signatures. Readers who are familiar with these subjects can skip this section.

\subsection{Reproducing Kernel Hilbert Space} \label{sec:RKHS}

\begin{definition}[Kernel] \label{def:kernel}
    Let $\Xt \subseteq \Rd$ be the input space with $d \in \N$, let $\V$ be some inner product space and let $\phi:\Xt \to \V$. Then we call the mapping $k: \Xt \times \Xt \ni (x, x') \mapsto \langle \phi(x), \phi(x') \rangle_\V \in \R$ a kernel.

\end{definition}
The primary advantage of using a kernel is that there is a way to compute the result of the inner product without having to explicitly compute the feature map $\phi(x)$.
This results in efficient computation despite the fact that the inner product space $\V$ may be of high dimension or even infinite dimensional.
This is known as the kernel trick (see \cite[Section 4]{steinwart2008support}).
Every kernel possesses a corresponding unique reproducing kernel Hilbert space (RKHS) (see \cite{muandet2016kernel}).

\begin{definition}[Reproducing kernel Hilbert space (RKHS)] \label{def:rkhs}
    The RKHS $(\Hb,\langle \cdot \; , \; \cdot \rangle_\Hb)$ of a kernel $k: \Xt \times \Xt \to \R$ is a Hilbert space of functions $f: \Xt \to \R$ with the following properties:
    \begin{enumerate}
        \item The kernel function with one fixed input is a function in the RKHS.
        \begin{equation}
            k(x, \cdot) \in \Hb \quad \text{for all} \; x \in \Xt.
        \end{equation}
        \item Reproducing property: the evaluation of a function in the RKHS at a point $x$ is equivalent to the inner product of the function with the kernel function evaluated at $x$.
        \begin{equation}
            \langle f, k(x, \cdot) \rangle_\Hb = f(x) \quad \text{for all} \; f \in \Hb, x \in \Xt.
        \end{equation}
    \end{enumerate}
\end{definition}

We can view the reproducing property as a mapping of each input $x$ to a function $k(x, \cdot)$ in the RKHS.
This feature map is sometimes called the canonical feature map.

A universal kernel is a kernel that corresponds to a RKHS that can approximate any continuous function arbitrarily well with respect to the uniform norm.
An example of a universal kernel is the Gaussian kernel defined by
\begin{equation}
    k(x, y) = \exp \left( -\frac{\Vert x - y \Vert^2}{2l^2} \right)
\end{equation}
where $l>0$ is the length scale parameter.
A universal kernel is also a characteristic kernel\footnote{To be precise, whether universality implies a characteristic kernel depends on the assumptions on the compactness of the input space, $\Xt$ and/or the choice of the kernel. If we assume $\Xt$ is compact then the assertion is true regardless of the choice of kernel. However, if the kernel is radial, such as the Gaussian kernel or rational quadratic kernel, then universality and characteristic are equivalent when $\Xt=\Rd$. Refer to \cite{sriperumbudur2011universality} for a comprehensive exposition.}.
\begin{definition}[Characteristic kernel] \label{def:characteristic_kernel}
    Given the set $\M$ of all probability measures on $\Xt$ equipped with the Borel $\sigma$-algebra, a kernel $k$ is characteristic if the map\footnote{We will use $E_{X \sim P}[\, \cdot \,]$ or $E_P[\, \cdot \,]$ to represent the integral $\int_\Xt \, \cdot \; dP(X)$ for brevity.} $\M \ni P \mapsto \mu_P := \int_\Xt k(X, \cdot) dP(X) = E_{X \sim P}[k(X, \cdot)] \in \Hb$ is injective. $\mu_P$ is called the kernel mean embedding of the probability measure $P$.
\end{definition}
This notion of a characteristic kernel is relevant to the relation between the MMD and the distributions being compared which we will explore next.

\subsection{Kernel Two Sample Test} \label{sec:kernel_two_sample_test}

A kernel two sample test is a non-parametric method based on the maximum mean discrepancy (MMD) to test whether two samples are drawn from the same distribution as introduced in \cite{gretton2012kernel}.
\begin{definition}[Maximum mean discrepancy (MMD)] \label{def:mmd}
    The MMD between two probability measures $P,Q \in \M$ is defined as the supremum of the difference between the expectation of the function evaluated on the two distributions where the supremum is taken over the unit ball of a RKHS $(\Hb,\langle \cdot \; , \; \cdot \rangle_\Hb)$ with respect to the induced norm $\Vert \cdot \Vert_\Hb=\sqrt{\langle \cdot \; , \; \cdot \rangle_\Hb}$.
    \begin{equation*}
        \text{MMD}(P, Q; \Hb) := \underset{f\in \Hb, \Vert f \Vert_\Hb \leq 1}{\sup} E_{X \sim P} [f(X)] - E_{Y \sim Q} [f(Y)]
    \end{equation*}
\end{definition}
The MMD is a distance metric between two probability measures.
Note the similarity with the Kantorovich-Rubinstein dual formulation of the Wasserstein distance, another distance metric over distributions.
If $\Hb$ is the space of all 1-Lipschitz functions then Definition \ref{def:mmd} is equivalent to the Wasserstein-1 or earth mover's distance.
The benefit of using an RKHS over 1-Lipschitz functions is that the supremum can be determined analytically.

\begin{proposition} \label{prop:mmd}
    Let $P,Q \in \M$ be two probability measures with $X \sim P$ and $Y \sim Q$.
    If $(\Hb,\langle \cdot \; , \; \cdot \rangle_\Hb)$ is the RKHS of a kernel $k$ and we assume Bochner integrability
    \begin{equation} \label{eq:bochner_int}
        E_{P}[\sqrt{k(X,X)}] < \infty \  \text{and} \  E_{Q}[\sqrt{k(Y,Y)}] < \infty
    \end{equation}
    then the MMD can be computed as the norm of the difference of the mean embeddings of the two probability measures.
    \begin{equation*}
        \text{MMD}(P, Q; \Hb) = \Vert \mu_{P} - \mu_{Q} \Vert_\Hb
    \end{equation*}
    where $\mu_{P}$ and $\mu_{Q}$ are as defined in Definition \ref{def:characteristic_kernel}.
\end{proposition}

\begin{proof}
    Firstly, by \cite[Lemma 3.1]{RN1023}, $\mu_{P}, \mu_{Q} \in \Hb$. Moreover
    \begin{flalign} \label{eq:mmd}
    \begin{split}
        \text{MMD}(P, Q; \Hb) &= \underset{f\in \Hb, \Vert f \Vert_\Hb \leq 1}{\sup} E_{P} [f(X)] - E_{Q} [f(Y)] \\
        &= \underset{f\in \Hb, \Vert f \Vert_\Hb \leq 1}{\sup} E_{P} [\langle f, k(X, \cdot) \rangle_\Hb] - E_{Q} [\langle f, k(Y, \cdot) \rangle_\Hb].
            \quad \text{(reproducing property)} \\
        \intertext{The expectation and the inner product can be interchanged due to the assumption in equation \eqref{eq:bochner_int}. Hence, we obtain}
        \text{MMD}(P, Q; \Hb) &= \underset{f\in \Hb, \Vert f \Vert_\Hb \leq 1}{\sup} \langle f, E_{P} [k(X, \cdot)] \rangle_\Hb - \langle f, E_{Q} [k(Y, \cdot)] \rangle_\Hb \\
        &= \underset{f\in \Hb, \Vert f \Vert_\Hb \leq 1}{\sup} \langle f, \mu_{P} - \mu_{Q} \rangle_\Hb. \\
        \intertext{By Cauchy-Schwarz inequality,
            $\langle f, \mu_{P} - \mu_{Q} \rangle_\Hb \leq \Vert f \Vert_\Hb \Vert \mu_{P} - \mu_{Q} \Vert_\Hb$ with equality only if $f$ and $\mu_{P} - \mu_{Q}$ are linearly dependent and since $\Vert f \Vert_\Hb \leq 1$, thus with $f := \frac{\mu_{P} - \mu_{Q}}{\Vert \mu_{P} - \mu_{Q} \Vert_\Hb}$, we have}
        \text{MMD}(P, Q; \Hb) &= \Vert \mu_{P} - \mu_{Q} \Vert_\Hb \\
        &= \sqrt{E_{P}~ [k(X,X)] - 2 E_{P,Q}~ [k(X,Y)] + E_{Q}~ [k(Y,Y)]}
    \end{split}
    \end{flalign}
\end{proof}

The Bochner integrability condition in equation \eqref{eq:bochner_int} holds for continuous bounded kernels or for continuous kernels on compact spaces (see \cite{sutherland2019unbiased}).
Equality \eqref{eq:mmd} now implies directly an estimator of the MMD.
\begin{proposition} \label{prop:unbiased_mmd}
    From \cite[Lemma 6]{gretton2012kernel}, let $X:=\{x_1,\ldots,x_m\}$ and $Y:=\{y_1,\ldots,y_m\}$ be samples from $P,Q \in \M$ respectively.
    Then an unbiased estimate of the $\text{MMD}^2$ is given by
    \begin{equation} \label{eq:mmd_est}
        \mmd(X, Y; k) := \underbrace{\frac{1}{m (m-1)} \sum^m_{i=1} \sum^m_{j=1,j\neq i} k(x_i, x_j)}_\text{A} -
            \underbrace{\frac{2}{m^2} \sum^m_{i=1} \sum^m_{j=1} k(x_i, y_j)}_\text{B} +
            \underbrace{\frac{1}{m (m-1)} \sum^m_{i=1} \sum^m_{j=1,j\neq i} k(y_i, y_j)}_\text{C}.
    \end{equation}
\end{proposition}

An alternative view is to see components A and C in equation \eqref{eq:mmd_est} as calculating the average similarity between samples in the same distribution and B in equation \eqref{eq:mmd_est} as calculating the average similarity between samples from different distributions
    since the kernel function is essentially calculating an inner product.
If the kernel function is characteristic, then the MMD is zero if and only if the two distributions are the same as shown in \cite{gretton2012kernel}.
The two sample kernel test is then to test whether the $\mmd$ is zero.
As the $\mmd$ can be computed using differentiable operations, it can serve as an objective function for gradient based optimisation methods.
For a more detailed exposition, we refer the reader to \cite{muandet2016kernel}.

\subsection{Signature of a Path} \label{sec:signature}

The kernel two sample tests described in Section \ref{sec:kernel_two_sample_test} are typically used for distributions over "static" i.e. non-causal values.
However, we are interested in distributions over paths and to capture the time causalities of these path values.
We can extend the test to distributions over paths by using the $\mmd$ defined in equation \eqref{eq:mmd_est} with $k$ being the signature kernel \cite{chevyrev2022signature}.
Before we define the signature of a path which is a sequence of iterated integrals, we use an abbreviated notation for the integral.
For a path $\x \in \D :=\{\x: [0,T] \to \R^d\}$ where at each time $t\in [0,T]$, $\x_t := (x_t^1, \ldots, x_t^d)$ is the value of the path at time $t$, the integral of the path from time $0$ to $t$ of asset $i \in \{1,\ldots,d\}$ is denoted as $S(\x)^i_{0,t}$.
\begin{equation} \label{eq:sin_int}
    S(\x)^i_{0,t} := \int_{0 < s < T} d\x^i_s = \x^i_T - \x^i_0.
\end{equation}

For any pair $i,j \in \{1, \ldots, d\}$, the double-iterated integral is denoted by $S(\x)^{i,j}_{0,T}$ and defined as
\begin{align} \label{eq:dbl_int}
    S(\x)^{i,j}_{0,T} &=\int_{0 < r < s < T} d\x^i_{r}\,d\x^j_{s} \\
    &= \int_{0 < s < T} S(\x)^i_{0,s} d\x^j_s.
\end{align}

This can be done recursively so that for any integer $k \geq 1$ and collection of indices $i_1, \ldots, i_k  \in \{ 1, \ldots, d\}$, we have
\begin{equation}
    S(\x)^{i_1,\ldots, i_k}_{0,T} = \int_{0< s < T} S(\x)^{i_1,\ldots, i_{k-1}}_{0,s} d\x^{i_k}_s.
\end{equation}

\begin{definition}[Signature of a path] \label{def:signature}
    Let $T \in \R$ be the time horizon and $\x \in \D$.
    The signature of path $\x$ over the time interval $[0,T]$ is the infinite sequence of real numbers which are the iterated integrals of $\x$.
    \begin{flalign*}
            \Sig(\x)_{0,T} := (1, S(\x)^1_{0,T}, \ldots, S(\x)^d_{0,T}, S(\x)^{1,1}_{0,T}, S(\x)^{1,2}_{0,T}, \ldots) \in \TV
    \end{flalign*}
    where the zeroth term is defined by convention to be 1 and the superscript runs through all possible combinations of indices $i_1, \ldots, i_k \in \{1, \ldots, d\}$ and $k \geq 1$.
\end{definition}

The signature of a path can be seen as a feature mapping from the path to an infinite dimensional feature space.
The following properties make the signature map a good candidate for a feature transformation for paths.:

\begin{enumerate}
    \item {\textbf{Uniqueness of the signature of a path}:
        The signature $\Sig(\x)_{0,T}$ is unique to $\x$ up to a tree-like equivalence.
        Refer to \cite{hambly2010uniqueness} for the definition of tree-like equivalence.
        Loosely speaking, there are 3 main instances that are invariant for the signature which correspond to instances where a path integral is invariant.
        First, a translation, i.e., $\x + C$ where $C$ is a constant value, would not change the value of the path integral since $d\x = d(\x + C)$.
        Second, a time reparameterisation, i.e., $\x_{\psi(t)}$ where $\psi: [0,T] \to [0,T]$ is a strictly increasing, continuous and surjective function, would not change the value of the path integral with a simple change of variables.
        Third, any retracements within the path, i.e., going from point $A$ to $B$ then back to $A$ via the exact same path, would not add any value to the path integral.
        These are all captured by the definition of a tree-like equivalence.
        }
    \item \textbf{Factorial decay of the signature terms} \cite[Proposition 2.2]{lyons2007differential}:
        Assuming $\x$ is a path of finite total variation. For each $1 \leq k \in \N$, the norm of the collection of terms of the signature of order $k$ (i.e. $k$ times iterated integrals) is bounded as follows
        \begin{equation}
            \left\lVert \left( S(\x)^{i_1,\ldots, i_k}_{0,T} \right)_{(i_1,\ldots,i_k) \in \{1,\ldots,d\}} \right\rVert_{(\R^d)^{\otimes k}} \leq \frac{\lVert \x \rVert_1^k}{k!}
        \end{equation}
        where $\lVert \x \rVert_1 := \underset{\substack{0=t_0 < \cdots <t_n=T \\ n \in \N}}{\sup} \sum_{i=1}^N\lvert \x_{t_{i}}-\x_{t_{i-1}}\rvert$ is the total variation of the path.
    \item \textbf{Universal non-linearity} \cite[Proposition 3]{fermanian2021embedding} :

        Let $D \subseteq \Pspc$ be a compact subset\footnote{Refer to \cite[Section 8]{lyons2014rough} on the topology of such a set of paths} of the space of paths with bounded total variation from $[0,T]$ to $\R^d$ such that for any $\x \in D$, $\x_0=\mathbf{0}$ and at least one dimension of $\x$ is a \emph{fixed} monotone coordinate (for uniqueness of the signature).
        Let $g : D \to \R$ be a continuous function. Then, for any $\epsilon > 0$, there exists a linear functional $h$ such that for all $\x \in D$,
        \begin{equation}
            \left\vert g(\x) - h(\Sig(\x)) \right\vert < \epsilon
        \end{equation}

        In other words, continuous functions on the path can be represented as a linear combination of the signature terms arbitrarily well.
\end{enumerate}

As the signature of a path is defined as a series of iterated integrals, it is the increments of the path that are significant in determining the signature.
For most downstream tasks related to financial time series, such as portfolio optimisation or risk management, it is the returns that are the important features rather than the absolute values of the prices.
This is the motivation behind using the log of prices as our path since increments of the log path are precisely the log returns.


{
The above properties also serve as the motivation for augmenting the path with a fixed time dimension, i.e., $\x \mapsto (t,\x)$.
A fixed time dimension anchors the path and ensures that the signature of each time augmented path is unique.
The universal non-linearity requires that at least one dimension of the path is monotone which is ensured by the augmentation of the time dimension.
}

Due to property 2, we can approximate the signature with a truncation at order $m$ and have a bound on the error between the truncated signature and the non-truncated signature.
\begin{definition}[Truncated signature]
    The truncated signature of order $m$ of a path $\x \in \Pspc$ is defined as
    \begin{equation*}
        \Sig^m(\x)_{0,T} := (1, S(\x)^1_{0,T}, \ldots, S(\x)^d_{0,T}, S(\x)^{1,1}_{0,T}, S(\x)^{1,2}_{0,T}, \ldots, S(\x)^{i_1,\ldots, i_m}_{0,T}) \in \TVm
    \end{equation*}
    where all combinations of indices $i_1, \ldots, i_m  \in \{ 1, \ldots, d\}$ are included.
\end{definition}

Indeed, due to property 2, we have for all $\x \in \Pspc, m \in \N$ that
\begin{equation}
    \left\lVert \left( S(\x)^{i_1,\ldots, i_k}_{0,T} \right)_{(i_1,\ldots,i_k) \in \{1,\ldots,d\}, k \geq m+1} \right\rVert_{\prod^{\infty}_{k=m+1}(\R^d)^{\otimes k}} \leq \sum^\infty_{k=m+1}\frac{\lVert \x \rVert_1^k}{k!}
\end{equation}

We refer the reader to \cite{lyons2022signature, chevyrev2016primer} for a more detailed exposition of the signature properties.

\subsection{Signature Kernel} \label{sec:signature_kernel}

We can construct a signature kernel as the inner product of the signatures of the two paths $\x$, $\y \in \D$ i.e. $\langle \Sig(\x), \Sig(\y) \rangle$.
This signature kernel is a universal kernel due to the universal non-linearity property\footnote{Recall that this property assumes that the path $\x \in D$ where $D \subseteq \D$ is a compact subset of the space of bounded variation paths from $[0,T]$ to $\R^d$ but it is possible to extend the universality to non-compact sets as shown in \cite{chevyrev2022signature}.} of the signature hence it is also characteristic.

It is also possible to first lift the paths into a feature space to enrich the path features before computing the signature of the augmented path.
The key insight of \cite{kiraly2019kernels} is that the kernel trick can be used to avoid explicitly computing the transformation when calculating the inner product of the signature of paths lifted into a RKHS.

\begin{definition}[Signature kernel] \label{def:signature_kernel}
    Let $(\Hb,\langle \cdot \; , \; \cdot \rangle_\Hb)$ be the RKHS of a kernel $k: \Xt \times \Xt \to \R$ and $\Sig^m$ be the truncated signature map of order $m$.
    The signature kernel is defined as the inner product $\D \times \D \ni (\x,\y) \mapsto \Sigkm(\x, \y) := \langle \Sig^m(k_\x), \Sig^m(k_\y) \rangle$ of the lifted paths $k_\x = (k_{\x_t})_{t\in [0,T]}$, $k_\y = (k_{\y_t})_{t\in [0,T]}\in \D_\Hb := \{\x: [0,T] \to \Hb \}$ i.e. the space of paths evolving in the RKHS $\Hb$ and $k_z=k(z, \cdot)$ for $z \in \Xt$.
\end{definition}

The kernel that is used to lift the paths is often called a static kernel and explicit computations of $\Sig(k_\x), \Sig(k_\y)\in \THm$ are avoided with only evaluations of the static kernel $k(\x_t,\y_t)$ required.
{
This can be done efficiently by approximating the signature kernel with discretised integrals as shown in \cite[Algorithm 3]{kiraly2019kernels}\footnote{{For an exact computation, there is also \cite[Algorithm 6]{kiraly2019kernels} for higher order computations. Choosing the order to be equal to the truncation level would provide the exact computation. We experimented with the higher order computations which requires significantly more memory but did not see meaningful improvements hence we use the 1st order approximations in our final results. There is also a method to compute an untruncated signature kernel by numerically solving a PDE where the solution is equal to the untruncated signature kernel introduced in \cite{salvi2021signature}. However, the MMD failed to converge in our experiments with this method.}}.
}
We primarily use the rational quadratic kernel as the static kernel.
It is a characteristic kernel and has been reported to perform well in practice (see e.g. \cite{binkowski2018demystifying}).
\begin{definition}[Rational quadratic kernel] \label{def:rational_quadratic_kernel}
    For $x,y \in \R^d$, the rational quadratic kernel is defined as
    \begin{equation*}
        k_{\text{quad}}(x, y) = \left( 1 + \frac{\Vert x - y \Vert^2}{2 \alpha l^2} \right)^{-\alpha}
    \end{equation*}
    where $\alpha > 0$ is the scale parameter and $l > 0$ is the length scale parameter.
\end{definition}
The rational quadratic kernel can be seen as a weighted sum of an infinite number of Gaussian kernels with different length scales where the squared inverse length scale follows a gamma distribution (see \cite[Section 4.2]{williams2006gaussian}).
The choice of the length scale parameter is important as it determines the smoothness of the kernel function.
It is a hyperparameter that needs to be tuned for the specific problem at hand.
The larger the length scale, the smoother the kernel function which means it is less sensitive to differences in the inputs.
However, if the chosen length scale is too small then even the smallest differences arising from samples from the same distribution will be overemphasised.

The main hyperparameter is the truncation order $m$ of the signature.
We modify the algorithm provided in \url{https://github.com/tgcsaba/KSig} for our specific use case.

\subsection{Sequences as Paths} \label{sec:sequences_as_paths}

The paths considered in the mathematics of signatures are continuous paths.
However, we only ever observe discrete sequences of prices in a financial time series.
Therefore, we need to convert the observed discrete sequence of prices into a continuous path with interpolation being the most obvious choice.

The two types of interpolation schemes typically paired with the use of signatures are linear and rectilinear interpolation which is discussed in \cite{fermanian2021embedding}.
Figure \ref{fig:interpolation} shows the difference between the two interpolation schemes applied to S\&P 500 index time series from 1-Jan-1995 to 15-Mar-1995.
We employ linear interpolation as it does not require any data transformations on the historical time series data and it demonstrates good performance in our experiments.

\begin{figure} [!ht]
    \centering
    \subfigure[Rectilinear interpolation]{
        \includegraphics[scale=0.4]{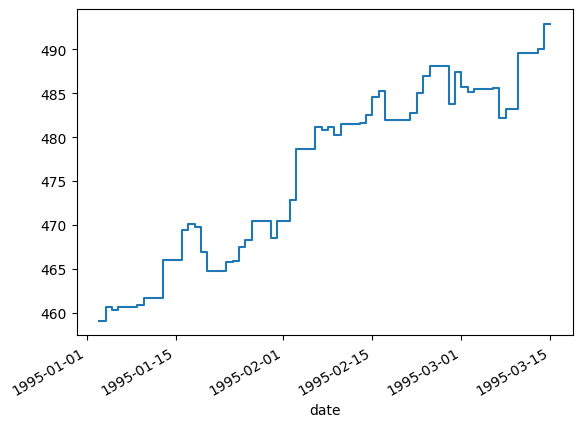}
        \label{fig:rectilinear}
    }
    \subfigure[Linear interpolation]{
        \includegraphics[scale=0.4]{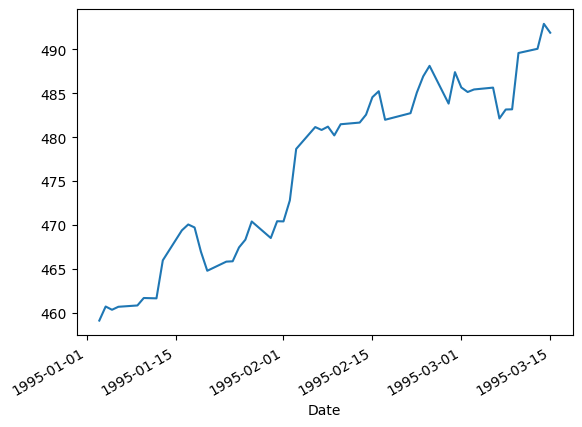}
        \label{fig:linear}
    }
    \caption{Interpolation schemes applied to data of the S\&P 500 index from 1-Jan-1995 to 15-Mar-1995}
    \label{fig:interpolation}
\end{figure}

This procedure can be seen as creating an embedding into the space of continuous functions as termed in \cite{fermanian2021embedding}.
Adding the time dimension to the sequence can be seen as enriching the embedding.
One particular augmentation that has reportedly good empirical performance (see e.g. \cite{fermanian2021embedding}) is the lead-lag augmentation introduced in \cite{chevyrev2016primer}.

\begin{definition} [Lead-lag and time augmentation] \label{def:lead-lag}
    Let $d,n \in \N$ and $\x = (\x_{t_0}, \x_{t_1}, \ldots, \x_{t_n}) \in \R^{d(n+1)}$ be a sequence of prices for $d$ assets at times $(t_0,\ldots,t_n) \in \R^{n+1}$.
    The lead-lag and time augmentation of $\x$ with lag $l \in \N$ is defined as $f^l_{\text{aug}}: \x \mapsto (\mathbf{t}, \x^{\text{lead}}, \x^{\text{lag}}) \in \R^{3d(n+1+l) + 1}$ and
    \begin{equation*}
        \mathbf{t}_i =
        \begin{cases}
            t_i & \text{if} \; i \leq n \\
            t_n & \text{if} \; i > n
        \end{cases} \quad \text{,} \quad
        \x^{\text{lead}}_i =
        \begin{cases}
            \x_{t_i} & \text{if} \; i \leq n \\
            \x_{t_n} & \text{if} \; i > n
        \end{cases} \quad \text{and} \quad
        \x^{\text{lag}}_i =
        \begin{cases}
            0 & \text{if} \; i < l \\
            \x_{t_{i-l}} & \text{if} \; i \geq l
        \end{cases}
    \end{equation*}
    for $i \in \{0, \ldots, n+l\}$.
\end{definition}

The lead-lag augmentation basically adds an additional dimension that lags behind the original sequence by a fixed lag i.e. concatenating a time shifted version of the original sequence.
This will require the end points of the original sequence to be extended and padded with the last value of the sequence.
To the best of our knowledge, there are no theoretical results explaining why the lead-lag augmentation is effective.
We speculate that it is due to the fact that the lead-lag augmentation allows the signature to capture information about the temporal dependencies of the sequence.

Once the generative model produces the log prices sequence, we augment it with the time sequence and the lead-lag sequence.
As we measure time in years, the time sequence is calculated based on the number of calendar days elapsed divided by 365.
For each sample sequence, we standardise both the log sequence of prices and the time sequence to start from zero by deducting the first value of the sequence from all the values of the sequence.
We add the lead-lag augmentation with a lag of 1 trading day to the sequence i.e. in the language of Definition \ref{def:lead-lag}, we set $l=1$.

\section{Generative Model} \label{sec:generative_model}

We present a generative model which is a neural network trained using the signature kernel based maximum mean discrepancy as the loss function\footnote{Refer to \cite{goodfellow2016deep} for an exposition on training neural networks}.
For the sake of illustration, we demonstrate in the rest of the paper with the price series of a single asset but the method can be extended to multiple assets with $\Xt \subseteq \Rd, d \geq 2$.
The model takes in three inputs based on a sequence, assumed to be a time series of a financial asset price, that is partitioned into $n$ time steps, $t_0, t_1, \ldots, t_n$ derived from the reference data. The three inputs at time $t_i$ are:
\begin{enumerate}
    \item Log return of the previous time step, $r_{t_{i-1}} \in \R$
    \item Time delta to the next time step (in calendar years), $\Delta t_i = t_i - t_{i-1} \in \R_+$
    \item Noise $z_{t_i} \in \R^{d_z}$ where $d_z \in \N$ is the dimension of the noise.
\end{enumerate}
Noise is a key ingredient of generative models that injects randomness into generated sequences, since the output would be deterministic without this component.
The importance of the two other inputs will be explored in Section \ref{sec:ablation}.

The output of the generative model is the log return of the current time step, $r_{t_i} \in \R$.
These returns are cumulatively summed to obtain the log price sequence.
Once we have the generated log sequence, we perform the lead-lag and time augmentation on the sequence.
With a batch of augmented sequences, we calculate the MMD between the generated sequences and the reference sequences from real data.
We employ variants of stochastic gradient descent (see e.g. \cite[Section 8.5]{goodfellow2016deep}) to minimize the MMD between the generated sequences and the reference sequences from real data.
Our objective is to have the generator output sequences that are close to the reference sequences, as measured by the MMD.

Using the MMD based on the signature kernel as a discriminator to train a generative model can be classified as a Score-based Generative Model (SGM) named by \cite{song2019generative}.
This is in contrast to Generative Adversarial Networks (GANs) where the discriminator also uses a neural network that is trained to distinguish between the real and generated samples as demonstrated in \cite{goodfellow2014generative}.
The adversarial nature of GANs are known to suffer from instabilities in training (see e.g. \cite{wiese2020quant,binkowski2018demystifying}).
We did not observe any instabilities during training using the non-adversarial SGM architecture.

\subsection{Neural Network Architecture} \label{sec:architecture}

At the heart of the generative model is a recurrent neural network that takes in the inputs $r_{t_{i-1}}$, $\Delta t_i$ and $z_{t_i}$ and outputs the log return of the current time step, $r_{t_i}$.
The specific neural network architecture starts with a single Long Short Term Memory (LSTM) layer, pioneered by \cite{hochreiter1997long}, with a hidden size of 64 neurons.

\begin{definition} [Long short term memory (LSTM) layer] \label{def:lstm}
    Let $d, n \in \N$, $\z = (\z_{t_1}, \ldots, \z_{t_n}) \in \R^{dn}$ where $\z_{t_i} \in \Rd$ is a column vector called the input at time $t_i$. With $h \in \N$ called the hidden size, we define a set of parameters $\theta:=\{ W_{i i}, W_{i f}, W_{i g}, W_{i o}, W_{h i}, W_{h f}, W_{h g}, W_{h o}, b_{i i}, b_{i f}, b_{i g}, b_{i o}, b_{h i}, b_{h f}, b_{h o} \}$ where $W_{i *} \in \R^{h \times d}, W_{h *} \in \R^{h \times h}$ are matrices called weights and $b_{i *}, b_{h *} \in \R^h$ are column vectors called biases. Let $c_{t_0}, h_{t_0} = (0,\ldots,0) \in \R^h$ be column vectors called the cell state and hidden state at time $t_0$ respectively. The LSTM layer is defined as
    \begin{equation} \label{eq:lstm}
        f_\theta: (\z_{t_i}, c_{t_{i-1}}, h_{t_{i-1}}) \mapsto (c_{t_i}, h_{t_i})
    \end{equation}
    based on the following operations:
    \begin{align}
        i_{t_i} &= \sigma\left(W_{i i} \z_{t_i}+b_{i i}+W_{h i} h_{t_{i-1}}+b_{h i}\right) \\
        f_{t_i} &= \sigma\left(W_{i f} \z_{t_i}+b_{i f}+W_{h f} h_{t_{i-1}}+b_{h f}\right) \\
        g_{t_i} &= \tanh \left(W_{i g} \z_{t_i}+b_{i g}+W_{h g} h_{t_{i-1}}+b_{h g}\right) \\
        o_{t_i} &= \sigma\left(W_{i o} \z_{t_i}+b_{i o}+W_{h o} h_{t_{i-1}}+b_{h o}\right) \\
        c_{t_i} &= f_t \odot c_{t_{i-1}}+i_t \odot g_t \\
        h_{t_i} &= o_t \odot \tanh \left(c_{t_i}\right)
    \end{align}
    where $\sigma$ is the sigmoid function\footnote{The sigmoid function is defined as $\sigma(x) := \frac{1}{1+e^{-x}}$} applied element-wise as is the tanh function, and $\odot$ is the element-wise product.
\end{definition}

The three inputs $r_{t_{i-1}}$, $\Delta t_i$ and $z_{t_i}$ are concatenated to form the input $\z_{t_i}$ for the LSTM layer, hence $d=2+d_z$ in the notation of Definition \ref{def:lstm}, in the case of a single asset.
The cell state and hidden state that is passed from one time step to the next can loosely be thought of as the long term and short term memory of the LSTM, hence the name.
The hidden state of the LSTM at time $t_i$, $h_{t_i}$, is then passed through a linear layer to produce the log return $r_{t_i}$.

\begin{definition} [Linear layer] \label{def:linear}
    Let $d,h \in \N$, $h_{t_i} \in \R^h$ be a column vector called the input at time $t_i$. We define a set of parameters $\psi := \{W \in \R^{d \times h}, b \in \R^d$\}. The linear layer is defined as
    \begin{equation} \label{eq:linear}
        f_\psi: h_{t_i} \mapsto r_{t_i} := W h_{t_i} + b \in \Rd.
    \end{equation}
\end{definition}

As we are working with a single asset, $d=1$ in the notation of Definition \ref{def:linear}.
The full architecture is illustrated in Figure \ref{fig:architecture}.

\begin{figure} [!ht]
    \centering
    \includegraphics[scale=0.4]{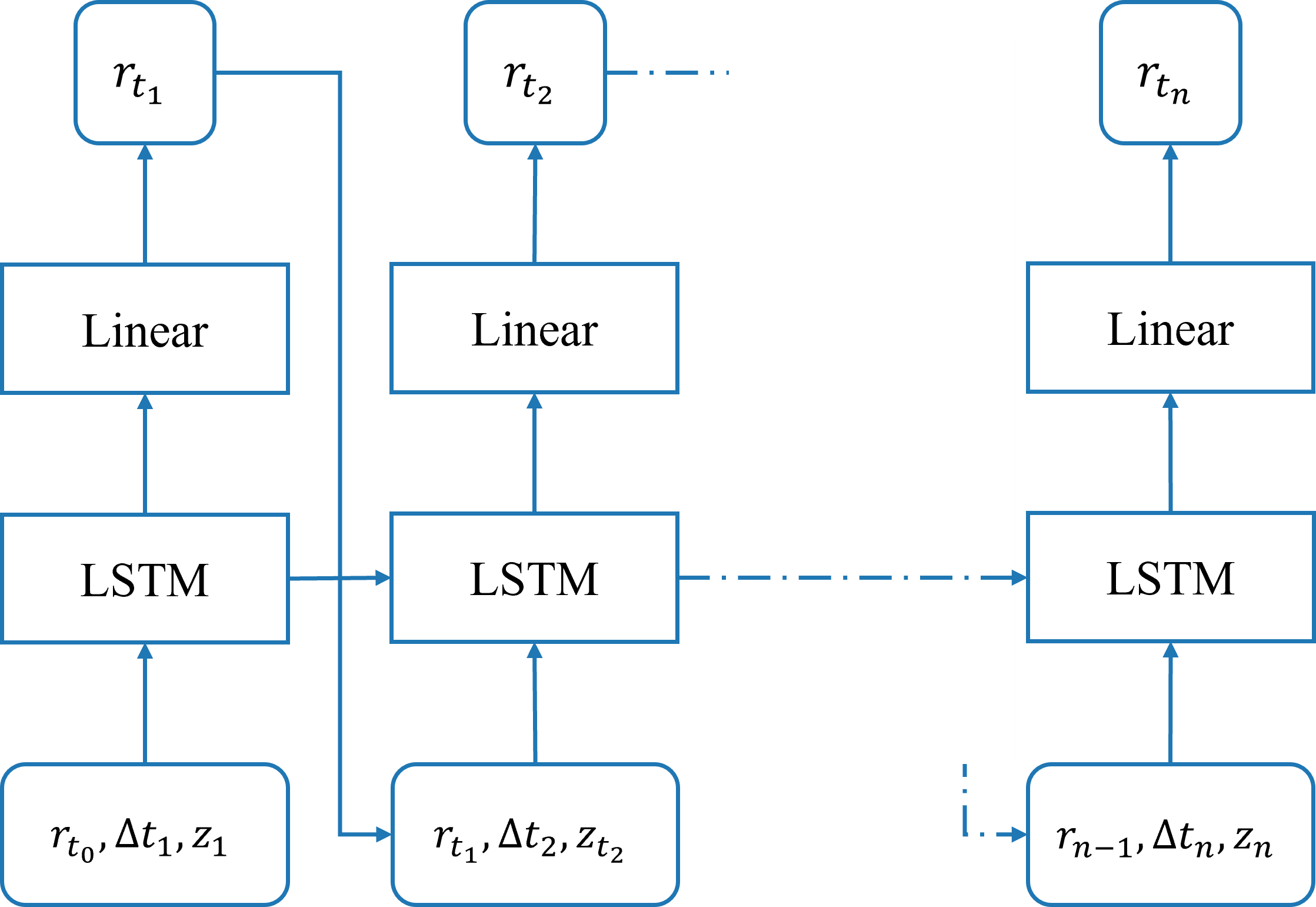}
    \caption{Neural network architecture}
    \label{fig:architecture}
\end{figure}

\subsection{The Cell State and Hidden State of the LSTM} \label{sec:conditioning}

For the first time step, the cell state and hidden state of the LSTM layer are initialised to zero.
Therefore, the first log return is generated without any influence from past returns.
However, we can also condition the cell state and hidden state with historical data before we start generating the sequence.

To illustrate using Figure \ref{fig:architecture}, we will calculate the first input $\Delta t_1$ from historical data, use $r_{t_0} = 0$ as a normalised value and feed it into the LSTM layer with the first noise input $z_{t_1}$.
The output of the linear layer will be the first generated log return $r_{t_1}$ but this will not be used as the input into the LSTM layer for the next time step nor will it be used in any downstream process.
Instead, we will continue to use the historical data to construct the inputs $r_{t_1}$ and $\Delta t_2$ for the next time step.
This will continue to a predetermined number of time steps before we start using the output as the return to be fed into the next time step and for use in any downstream process.
In this case, the cell state and hidden state of the LSTM will be conditioned on the historical data preceding the generated sequence.
The importance of this step will be explored in Section \ref{sec:ablation}.

Once we have collected all the outputs $r_{t_i}$, we cumulatively sum them to obtain the generated log price sequence $\x = (x_{t_0}, x_{t_1}, \ldots, x_{t_n})$ where $x_{t_i} = \sum_{j=0}^i r_{t_j}$ with $r_{t_0}$ normalised to zero.
If we used $k \in \N$ historical log price data points to condition the LSTM, then the first $k$ values of the generated sequence will be the same as the historical data.
To prevent an artificial advantage to the generative model, we will only use the log returns after the conditioning phase for both the generated and reference data sequence i.e. we will discard the first $k$ generated log returns to get
\begin{equation} \label{eq:trunc}
    \x_{\text{trun}} = (x_{t_{k}}, x_{t_{k+1}}, \ldots, x_{t_n}).
\end{equation}
In our experiments reported in Section \ref{sec:experiments}, we use $k=50$.

We perform the lead-lag and time augmentation on the generated log price sequence to obtain $\tilde{\x} = f^1_{\text{aug}}(\x_{\text{trun}})$.
The same augmentation is performed on the reference data log price sequence to obtain $\tilde{\y}$.
A batch of augmented generated sequences $\mathbf{X} := \{\tilde{\x}_1, \ldots, \tilde{\x}_b\}$ and augmented reference sequences $\mathbf{Y := \{\tilde{\y}_1, \ldots, \tilde{\y}_b\}}$, where $b$ is the batch size, are then used to calculate the unbiased MMD estimator using the signature kernel i.e. $\mmd(\mathbf{X}, \mathbf{Y}; \Sigkm)$.

\begin{remark}
    We emphasize that the conditioning described above is different from generating a conditional distribution such as in \cite[Section 5]{liao2024sig} and \cite[Section 3.4]{issa2023non}.
    In both cases, the generator is trained to reproduce the conditional distribution of the future prices (or returns) given a \textbf{specific} past sequence.
    In \cite{liao2024sig}, this is done by first assuming the true conditional expectation of the future price signature is a linear function of the past price signature and then using linear regression to estimate the parameters.
    Thereafter, the Monte Carlo estimate of the $l^2$ norm between the generator's prediction given the past price serves as the loss function to train the generator.
    Using a linear estimator for the future path signature is no longer possible when we first lift the price path into a RKHS such as in our approach and in \cite{issa2023non}.
    Therefore, the approach in \cite{issa2023non} is to use a \textbf{single} historical sample of past and future sequence to estimate the score of multiple generated sequences conditioned on the same past sequence.
    In our approach, we estimate the score (MMD) by comparing a batch of generated sequences conditioned on a \textbf{range} of past sequences to the batch of future sequences resulting from the set of past sequences.
    We found that this approach combined with our choice of noise distribution described in Section \ref{sec:noise} works well when generating long sequences (equal to about 1 calendar year of trading days) as we will show in Section \ref{sec:experiments}.
\end{remark}

\subsection{Noise} \label{sec:noise}

Noise is a key input to the generative model as it is the source of the stochasticity.
We can think of the generative model as mapping from the (potentially higher dimensional) noise distribution to the distribution of the returns.
If we model the noise by a distribution that is closer to the distribution of the returns, then it should be easier for the model to learn suitable parameters $\theta$ and $\psi$.
We test this conjecture with a toy example in Section \ref{sec:noise_dist}.
For the rest of the experiments in Section \ref{sec:experiments}, we use a moving average model to simulate the noise.

\begin{definition}[Moving average model] \label{def:ma_model}
    For $\omega,\beta_1,\ldots,\beta_p \in \R$ and $p \in \N$, the moving average model of order $p$ (MA$(p)$) is defined as
    \begin{flalign*}
        z_{t} &= \sigma_{t} \epsilon_t \\
        \sigma_{t}^{2} &= \omega + \sum_{i=1}^{p} \beta_i z_{t-i}^{2} \\
    \end{flalign*}
    where $z_t$ is the target variable and $\epsilon_t$ is a standard normal random variable.
\end{definition}

The MA model assumes that the target variable, which are the log returns, follows a normal distribution with a conditional variance that is a function of the past squared returns.
However, the distribution of returns are known to have fatter tails and taller peaks than the normal distribution.
Therefore, we would first like to transform the returns to a distribution closer to the normal distribution before fitting the MA model\footnote{An alternative is to choose a different distribution for $z_t$ in Definition \ref{def:ma_model} that is closer to the distribution of the returns. However, this increases the complexity in fitting the model and the choice of distribution is not obvious.}.
One way to do this is to use the \textit{inverse} Lambert W transformation as demonstrated in \cite{goerg2015lambert}.

\begin{definition}[Lambert W transformation] \label{def:lambert_w}
    Let $\delta \in \R$, $R$ be a $\R$ valued random variable with mean $\mu_R$ and standard deviation $\sigma_R$ and $U:=\frac{R-\mu_R}{\sigma_R}$. The Lambert W transformation of $R$ is defined as $W(R):= U \exp{\left(\frac{\delta}{2}U^2\right)} \sigma_R + \mu_R$.
\end{definition}

When $\delta >0$, the Lambert W transformation has the effect of reshaping the distribution of $R$ to have fatter tails.
This means that $W^{-1}$, the inverse function of $W$, that we call the \textit{inverse} Lambert W transformation will reshape the distribution of the returns to have lighter tails.
We can therefore use the \textit{inverse} Lambert W transformation to reshape the returns to have a distribution closer to the normal distribution before fitting the MA model.

Before performing the \textit{inverse} Lambert W transformation, we perform a preprocessing step for numerical reasons explained in \cite{lecun2002efficient}.
First, we annualise the log returns by dividing them by the time elapsed in years.
Then we normalise them to have zero mean and unit variance to get the annualised and normalised returns, $r_Z$.
The \textit{inverse} Lambert W transformation is then applied by estimating the parameters $\delta, \mu_R$ and $\sigma_R$ using iterative generalised method of moments\footnote{Although $r_Z$ has been normalised with zero mean and unit variance{,} the estimation process will aim to produce a kurtosis of 3 for the transformed variables which will cause the value of $\mu_R$ and $\sigma_R$ to deviate from zero.} from \cite[Section 4.2]{goerg2015lambert}.
Using the estimated parameters, we obtain the "Gaussianised" returns, {$r_W = W^{-1}(r_Z)$}.

The parameters $\omega$ and $\beta_1,\ldots,\beta_p$ of the MA model are then fitted using maximum likelihood estimation with $r_W$ as the target variable.
When simulating noise for the generative model, we will sample a date from the historical data and label it $t_0$.
The prior $p-1$ dates from the historical data will be labeled $t_{1-p}, \ldots, t_{-1}$, where $p$ is the order of the MA model.

{
\begin{definition}[Noise generation with MA model] \label{def:ma_model_noise}
    Let the noise at time $t_i$ be denoted by $z_{t_i} \in \R^{d_z}$ where $d_z \in \N$ is the dimension of the noise.
    Let $\Hz(i) = (z_{t_{i+1-p}},\ldots,z_{t_i}) \in \R^{d_z p}$ be the realised noise sequence from time $t_{i+1-p}$ to $t_i$ for $i \in \{1, \ldots, n\}$ and $p \in \N$ be the order of the MA model  from Definition~\ref{def:ma_model}.
    Specifically, if $z_{j,t_i}$ is the $j$-th component of $z_{t_i}$, then let $z_{j,t_i} = r_{W,t_i}$ for $i \in \{1-p, \ldots, 0\}$, $j \in \{1,\ldots,d_z\}$ where $r_{W,t_i} \in \R$ is the Gaussianised return at time $t_i$.
    Given parameters $\eta = \{\omega,\beta_1,\ldots,\beta_p\}$ of the MA model from Definition~\ref{def:ma_model}, we define the noise at time $t_i$ via the map $f_{\eta}$ given by
    \begin{equation} \label{eq:noise}
        \R^{d_z p} \ni (z_{t_{i-p}},\ldots,z_{t_{i-1}})=\Hz(i-1) \mapsto  f_\eta(\Hz(i-1)):= z_{t_i} \ \text{for} \  i \in \{1, \ldots, n\}
    \end{equation}
    which is explicitly defined through the following operations:
    \begin{align}
        \sigma^2_{j,{t_i}} &= \omega + \sum_{k=1}^{p} \beta_k z_{j,t_{i-k}}^2 \\
        z_{j,t_i} &= \sigma_{j,t_i} \epsilon_{j,t_i}
    \end{align}
    where $\epsilon_{t_i}$ is a standard normal random variable.
\end{definition}
In other words, for the first $p$ time steps, the historical data is used in the generative process, thereafter, the noise is generated using the MA model based on the previous $p$ noise values.
In our experiments reported in Section \ref{sec:experiments}, we use $p=20$.
}

In summary, the goal of using a MA model, fitted in this manner, is to produce noise that is closer to the distribution of the historical returns than i.i.d.\,Gaussian noise.
{
Note that we do not reverse the transformations, which include inverse Lambert W, annualisation and normalisation, to the sampled noise $z_{j,t_i}$ as we found that this did not work empirically.
We speculate that this is due to the reverse process producing returns with a much more extreme fat tail than observed in the original data.
Therefore, we leave the more difficult task of mapping from the noise to the empirical log returns to be done by the neural network.
}

\subsection{Training Algorithm} \label{sec:training_algorithm}

With all the components in place, we can now describe the training algorithm for the generative model in Algorithm \ref{alg:training}.
Details of the implementation can be found in \url{https://github.com/Anonymous1701/mmd-repo}.

\begin{algorithm} [!ht]
    \caption{Training algorithm} \label{alg:training}
    \KwIn{Initial generative model parameters $\Theta = \{\theta,\psi\}$, learning rate $\alpha$, MA model parameters $\eta$ {and $p$}, Dataset $\Data$, number of epochs $N$, batch size $B$, truncation level $m$, number of historical data points for the conditioning phase $k$, sample length $(n+1)$ and noise dimension $d_z$}
    \KwOut{Trained generative model parameters $\Theta$}
    \For{$n=1$ \KwTo $N$}{
        \For{$j=1$ \KwTo $\lfloor \lvert\Data\rvert / B \rfloor$}{
            \For {$b=1$ \KwTo $B$}{
                Sample a date from $\Data$ without replacement and label it as $t_0$; \\
                Retrieve the historical data {$\Hz(0)=(r_{W,t_{1-p}},\ldots,r_{W,t_0})$} and $\mathcal{H}=(r_{\Data,t_1},\ldots,r_{\Data,t_{k-1}},\Delta t_1,\ldots,\Delta t_n)$ from $\Data$; \\
                Sample a batch of noise {$z_{t_i} = f_\eta(\Hz(i-1))$} for $i \in \{1, \ldots, n\}$ according to equation \eqref{eq:noise}; \\
                Generate a sequence of log returns $r_{\Theta,t_i} = f_\psi(h_{t_i})$ and $h_{t_i} = f_\theta(r_{t_{i-1}}, \Delta t_i, z_{t_i})$ for $i \in \{1, \ldots, n\}$ by applying equations \eqref{eq:lstm} and \eqref{eq:linear}; \\
                Cumulatively sum $r_{\Theta,t_i},r_{\Data,t_i}$ for $i \in \{0,\ldots,n\}$ with $r_{\Theta,t_0} = r_{\Data,t_0} = 0$ to get the historical and generated log price sequence $\x_\Theta, \y$ respectively; \\
                Truncate first $k$ values for both sequences to get $\x_{\Theta,\text{trun}}, \y_{\text{trun}}$ according to equation \eqref{eq:trunc}; \\
                Perform lead-lag augmentation to get the sequences $\tilde{\x}_{\Theta,b} = f^1_{\text{aug}}(\x_{\Theta,\text{trun}})$ and $\tilde{\y} = f^1_{\text{aug}}(\y_{\text{trun}})$; \\
            }
            Collect batch of samples ${\mathbf{X}_\Theta = \{\tilde{\x}_{\Theta,1}, \ldots, \tilde{\x}_{\Theta,B}\}, \mathbf{Y} = \{\tilde{\y}_1, \ldots, \tilde{\y}_B\}}$; \\
            Compute $\mmd(\mathbf{X}_\Theta, \mathbf{Y}; \Sigkm)$; \\
            Obtain the gradient $\Delta_\Theta \leftarrow \nabla_\Theta \mmd(\mathbf{X}_\Theta, \mathbf{Y}; \Sigkm)$; \\
            Use preferred variant of stochastic gradient descent to update $\Theta \leftarrow \Theta - \alpha \Delta_\Theta$; \\
        }
        Perform early stopping and learning rate decay if required based on an average of $\mmd(\mathbf{X}_\Theta, \mathbf{Y}; \Sigkm)$ values; \\
    }
\end{algorithm}

\section{Experiments} \label{sec:experiments}

In this section, we describe the experiments conducted to train the generative model using Algorithm \ref{alg:training} and the results obtained.
There are mainly four parts to the experiments.
Firstly, we demonstrate that using noise that is closer to the distribution of the returns improves the performance of the generative model.
Secondly, we detail the hyperparameters for the experiment focused on a single asset, the S\&P 500 index and show that the trained generative model can produce sequences exhibiting behaviours according to stylised facts of financial time series.
Thirdly, the generative model is used to train a reinforcement learning agent to demonstrate an application for a portfolio optimisation task.
Finally, we show how to improve the robustness of the generative model by altering the noise distribution to different market regimes.

\subsection{Data Preprocessing} \label{sec:data_preprocessing}

We use the daily closing level of the S\&P 500 index from 1 Jan 1995 to 28 Dec 2023 sourced from Yahoo Finance.
In choosing the data period, there is a tradeoff between having enough data and ensuring the distribution is not too different from the present.
Although the S\&P 500 index has a much longer history, we chose 1995 as a cutoff, around the period when electronic trading was becoming more prevalent and the distribution of returns likely changed as a result (see e.g. \cite{freund1997market}).
We also withhold data from 19 Sep 2018 to 28 Dec 2023 as a test set for the portfolio optimisation task.

We aim to generate sequences with lengths of 250 data points, roughly equivalent to one year of data, and therefore capture the seasonal effects of the market.
In addition, we condition the neural network with 50 data points of historical data.
For clarity, this means $k=50$ and $n=299$ using the notation of equation \eqref{eq:trunc}, so that after truncation in line 9 of in Algorithm \ref{alg:training}, we have a generated log price sequence $\{\x_{t_{50}},\ldots,\x_{t_{299}}\}$ of length 250.

If we extract non-overlapping sequences from the historical data, there would be insufficient data to train the model.
Therefore, we use overlapping sequences with a stride of 50 data points.
Figure \ref{fig:stride} shows an example of overlapping sequences with a stride of 2 and length of 8.

Increasing the stride will reduce the overlapping between sequences but will also reduce the number of sequences available for training.
The choice of stride aims to balance the tradeoff between having enough data to train the model while ensuring the sequences have sufficient differences between them to effectively train the model.
This type of sampling will create biases in the training data but we leave the exploration of this issue to future work.

\begin{figure} [!ht]
    \centering
    \includegraphics[scale=0.5]{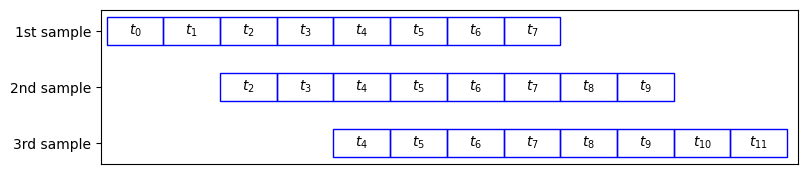}
    \caption{Example of overlapping sequences with a stride of 2 and length of 8}
    \label{fig:stride}
\end{figure}

\subsection{Noise Distribution} \label{sec:noise_dist}
We test the conjecture that using noise which is closer to the distribution of the returns will improve the performance of the generative model with a toy example utilising the Heston stochastic volatility model introduced in \cite{heston1993closed}.
We generate training data using the Heston model and fit two generative models.
The first model uses i.i.d.\,standard Gaussian noise and the second model uses zero mean Gaussian noise with stochastic variance\footnote{{The reader may wonder why we did not use the noise process described in Section~\ref{sec:noise}. 
The main goal is to provide some empirical evidence that if the noise input for the generative model is closer to the actual distribution of the paths, it should produce a better result. 
In this particular case, we chose the noise distribution to follow the CIR process for the variance of the price process since it is the actual "noise" process of the Heston model.}} based on the same parameters as the simulation but independently generated.
Details of the experiment parameters can be found in Appendix \ref{app:heston}.

Post-training, we generate 1,000 samples from each model and perform permutation based two sample tests using the unbiased statistic $\mmd$.
The null hypothesis is that the two samples are drawn from the same distribution.
We use 10,000 permutations to estimate the p-value.
The results are shown in Table \ref{tab:heston_mmd} and are consistent with the hypothesis that we can improve the generative model by using a noise distribution that is closer to the distribution of the returns.

\begin{table} [!ht]
    \centering
    \begin{tabular}{lrr}
        \toprule
        Model & $\mmd$ & p-value \\
        \midrule
        Stochastic volatility noise & 0.043 & 0.19 \\
        I.i.d.\,Gaussian noise & 1.434 & 0.0 \tablefootnote{The test statistic is significantly higher than all other permutation samples.} \\
        \bottomrule
    \end{tabular}
    \caption{Hypothesis testing using 10,000 permutations}
    \label{tab:heston_mmd}
\end{table}

Using i.i.d.\,Gaussian noise does not generate returns exhibiting volatility clustering which is a stylised fact of financial time series (see e.g.\cite{cont2001empirical}).
Indeed, the S\&P 500 index log returns in Figure \ref{fig:spx_vol_clustering} exhibit clear volatility clustering behaviour while the i.i.d.\,Gaussian noise generated log returns in Figure \ref{fig:lack_vol_clustering} do not.
In order to create a volatility clustering effect on the Gaussian noise sequence, we model the variance of the noise as a moving average (MA) of the squared returns.
This allows the generative model to produce returns that exhibit volatility clustering as seen in Figure \ref{fig:ma_vol_clustering}.

\begin{figure} [!ht]
    \centering
    \subfigure[S\&P 500 log returns]{
        \includegraphics[scale=0.33]{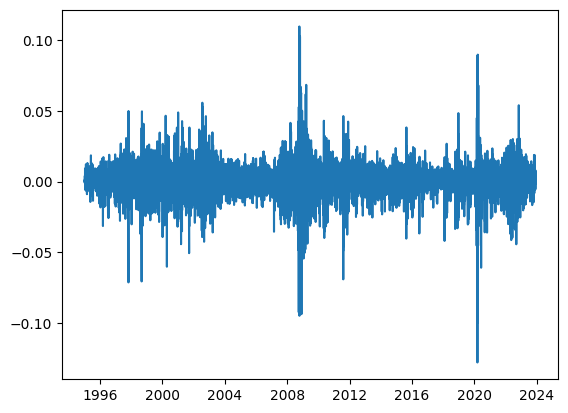}
        \label{fig:spx_vol_clustering}
    }
    \subfigure[Gaussian noise generated log returns]{
        \includegraphics[scale=0.33]{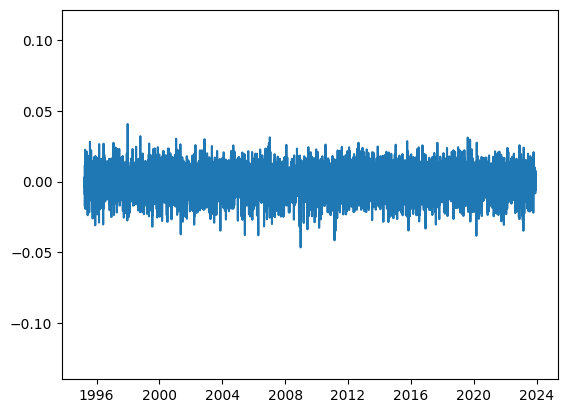}
        \label{fig:lack_vol_clustering}
    }
    \subfigure[MA(20) noise generated log returns]
    {
        \includegraphics[scale=0.33]{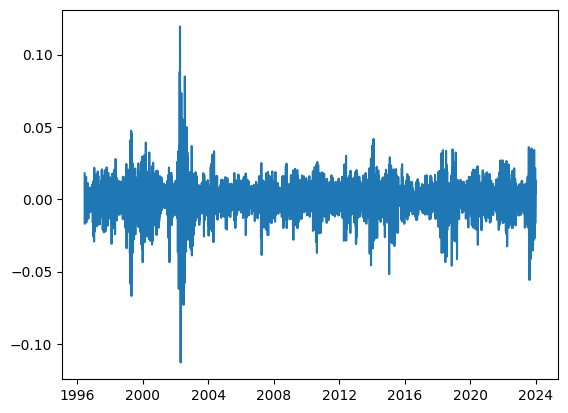}
        \label{fig:ma_vol_clustering}
    }
    \caption{Volatility clustering in generated log returns}
    \label{fig:vol_clustering}
\end{figure}

\subsection{Hyperparameters} \label{sec:hyperparameters}

\subsubsection*{MA Model}
We fit the MA model with maximum likelihood estimation using the python package \texttt{arch} \cite{sheppard2024bashtage/arch}.
We found empirically by a hyperparameter search that the MA$(20)$ model i.e. 20 lags produces good results.

\subsubsection*{Rational Quadratic Kernel}
There are two main hyperparameters for the rational quadratic kernel from Definition \ref{def:rational_quadratic_kernel} which are $\alpha$, the shape parameter for the Gamma distribution and $l$, the length scale parameter
We used a default value of $\alpha=1$ and experimented with different values of $l$ ending with a value of $l=0.1$ which produced the best results.

\subsubsection*{Truncated Signature}
We experimented by progressively increasing the truncation order of the signature kernel $m$ from 5 to 10.
Due to the long sequences, increasing the order further still has a significant impact on the absolute value of the $\mmd$.
However, we did not see a material improvement in the quality of the generated sequences when we progressively increased the order beyond 10 which is why we used $m=10$.

\subsubsection*{Noise}
For the noise input, there is a choice for the dimension of the noise $d_z$.
{
We investigated $d_z=1,2,4,6,8$ and found that $d_z=4$ appears to be a minimum to achieve a good result in terms of convergence in the $\mmd$.}
The better performance when using a noise dimension of 4 could be due to a larger number of parameters in the model since a larger input size creates more parameters in the first layer of the neural network.
The larger number of parameters could allow the model to better capture the complexity of the data.
Therefore, to create a level playing field, we also experimented with quadrupling the hidden size {for $d_z=1,2,4$} to see if the performance gap between the noise dimensions would be reduced.

Table \ref{tab:noisedim_hidden} shows the results of hypothesis testing with 100 generated samples, using the same permutation based two sample test as in {Section \ref{sec:noise_dist}} for different noise dimensions and hidden sizes.
While enlarging the model size does improve the performance, there is still a significant performance gap between $d_z=4$ and {the smaller values.}
Enlarging the model with $d_z=4$ {or increasing the number of noise dimensions} does not have any meaningful improvement in performance hence we choose the model with $d_z=4$ and hidden size of 64 as the final model.

\begin{table} [!ht]
    \centering
    \begin{tabular}{ccrr}
        \toprule
        $d_z$ & Hidden size & $\mmd$ & p-value \\
        \midrule
        1 & 64 & 74.626 & 0.0 \\
        1 & 256 & 1.724 & 0.087 \\
        2 & 64 & 35.522 & 0.0 \\
        2 & 256 & 4.351 & 0.008 \\
        \textbf{4} & \textbf{64} & \textbf{-0.578}\tablefootnote{$\mmd$ is an unbiased estimate of a random variable that is greater than or equal to zero. Therefore, if the true value is zero, some of the sample values will likely be negative.} & \textbf{0.628} \\
        4 & 256 & 0.291 & 0.345 \\
        6 & 64 & 1.444 & 0.127 \\
        8 & 64 & -0.443 & 0.577 \\
        \bottomrule
    \end{tabular}
    \caption{Performance of different noise dimensions and hidden sizes. The $\mmd$ and p-value are calculated using 10,000 permutations.}
    \label{tab:noisedim_hidden}
\end{table}

Post-training analysis of this model shows that one of the four noise dimensions has a dominant effect on the generated sequence.
Figure \ref{fig:noise_dim_impact} shows the effect of setting all but one of the noise dimensions to zero when generating new sequences with the trained generator.
The blue lines are the log returns of the sequences generated with no alterations to the noise while the orange lines are the log returns of the sequences generated with all but one of the noise dimensions set to zero.
Figure \ref{fig:hist_error} shows the histogram of the \textit{difference} in log returns between the sequence generated with unaltered noise and the sequence generated with all but the one noise dimension set to zero.
Both figures show that even with all but the 3rd noise dimension set to zero, we still see the generated sequence of log returns very close to the original sequence.
However, even though the contribution of the non-dominant noise dimensions looks marginal, we cannot omit them and reduce to a one-dimensional model without worsening the $\mmd$ drastically.

\begin{figure} [!ht]
    \centering
    \includegraphics[scale=0.4]{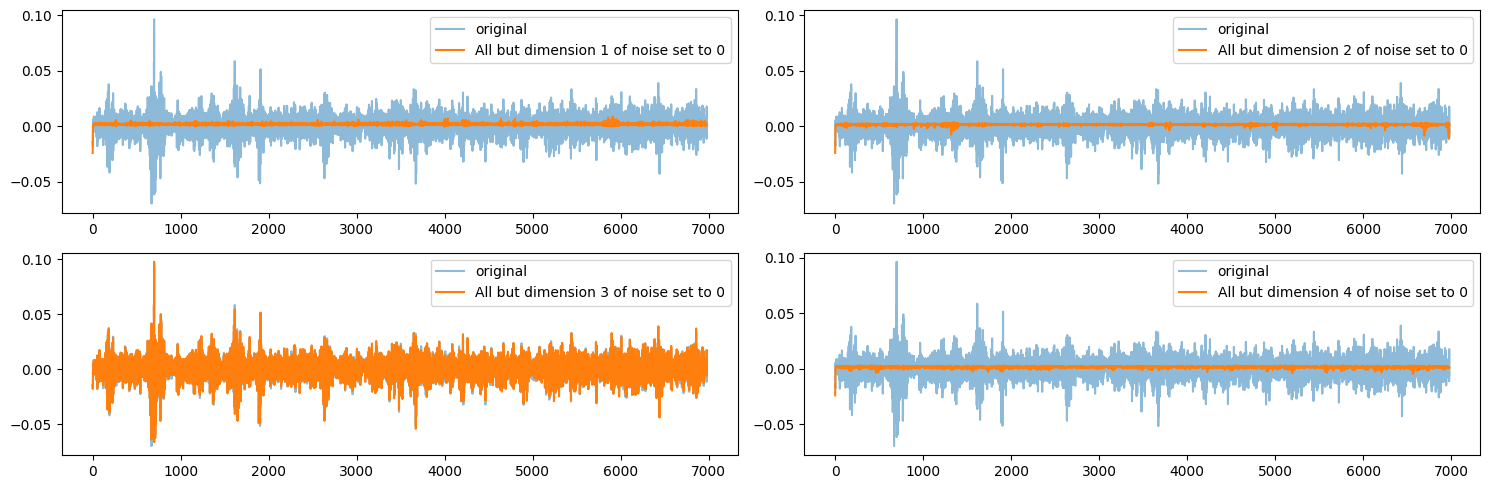}
    \caption{Impact of noise dimension on generated log returns sequence. For each orange plot, we set all but one noise dimension to zero when using the model to output the log return. The blue plots are the same across all four plots which is produced with the unaltered noise. The y-axes are the returns level and the x-axes are the time steps.}
    \label{fig:noise_dim_impact}
\end{figure}

\begin{figure} [!ht]
    \centering
    \subfigure[Dim 3 vs Dim 1]{
        \includegraphics[scale=0.34]{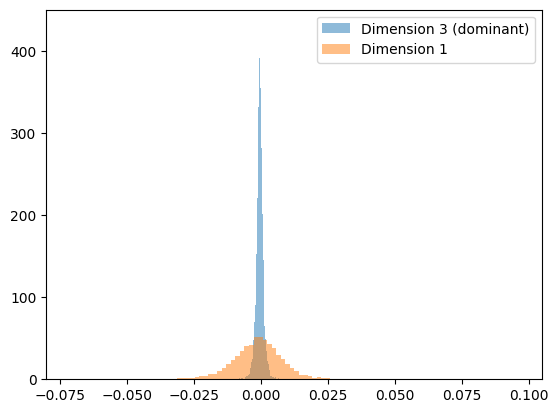}
        \label{fig:hist_3_1}
    }
    \subfigure[Dim 3 vs Dim 2]{
        \includegraphics[scale=0.34]{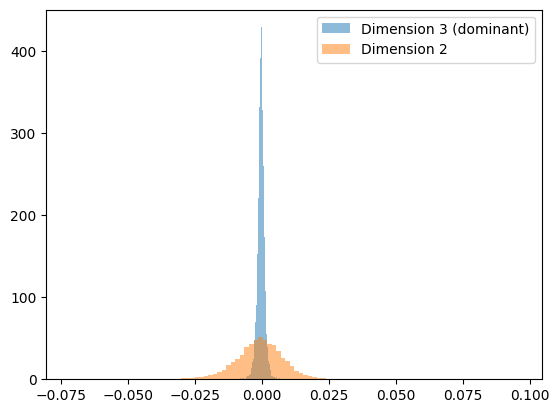}
        \label{fig:hist_3_2}
    }
    \subfigure[Dim 3 vs Dim 4]{
        \includegraphics[scale=0.34]{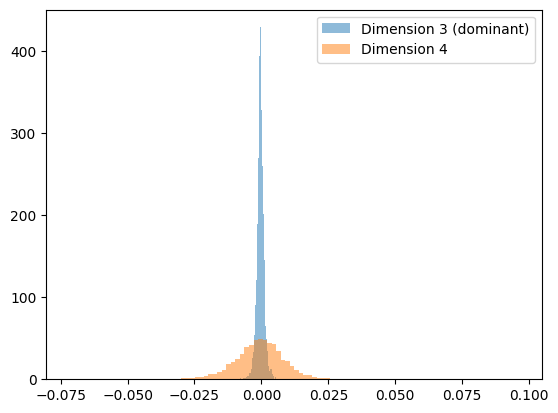}
        \label{fig:hist_3_4}
    }
    \caption{Histogram of the \textit{difference} in log return between the sequence generated with unaltered noise and the sequence generated with all but the one noise dimension set to zero. The blue plot is produced with the dominant noise dimension i.e. the noise dimension that produces the log returns closest to the unaltered noise while the orange plots are of the non-dominant dimensions.}
    \label{fig:hist_error}
\end{figure}

{
We demonstrate the impact on the $\mmd$ by generating 100 samples from the trained model.
For each sample, we generate an alternative sample where we set all but the dominant noise dimension to zero.
We perform the same permutation based two sample tests as in Section \ref{sec:noise_dist} to compare the two sets of generated sequences with the original S\&P 500 sequences.
It is clear from Table~\ref{tab:noise_dim_impact} that despite the seemingly small differences in the generated log returns when all but the dominant noise dimension is set to zero, it has a significant impact on the $\mmd$ measured against the S\&P 500 samples.
}

\begin{table} [!ht]
    \centering
    {
    \begin{tabular}{lrr}
        \toprule
        Sample Type & $\mmd$ & p-value \\
        \midrule
        Unaltered noise & -0.578 & 0.628 \\
        Dominant dim & 12.239 & 0.0 \\
        \bottomrule
    \end{tabular}
    }
    \caption[Impact of noise on $\mmd$]{Impact of noise on the $\mmd$ with 10,000 permutations. The unaltered noise is the noise used to train the model while the dominant dim is the noise with all but the dominant noise dimension set to zero.}
    \label{tab:noise_dim_impact}
\end{table}

\subsection{Ablation Study} \label{sec:ablation}

We conduct an ablation study on the importance of the different inputs to the model.
Specifically, we try removing the two non-stochastic inputs which are the return of the previous time step, $r_{t_{i-1}}$ and the time delta to the next time step, $\Delta t_i$.
In addition to the base model, we train 3 additional models with the following modifications to the inputs:
\begin{enumerate}
    \item Time delta to the next time step i.e. $\Delta t_i$ is omitted.
    \item The return of the previous time step i.e. $r_{t_{i-1}}$ is omitted.
    \item Both the return of the previous time step and time delta to the next time step are omitted.
\end{enumerate}
We perform the same permutation based two sample tests as in {Section \ref{sec:noise_dist}} to compare the generated sequences with the original S\&P 500 sequences.
All sequences are generated with the same noise input and the tests are conducted with the same random seed.
Table \ref{tab:ablation_mmd} shows the $\mmd$ and p-value of the two sample tests.
While it is no surprise that removing $r_{t_{i-1}}$ has a significant adverse effect on the performance showcasing the importance of $r_{t_{i-1}}$ for the generation of the time series, it is interesting to note that removing $\Delta t_i$ also has a large adverse effect.
This suggests that there are important time scale related features in the data that the model is able to learn and exploit.

\begin{table} [!ht]
    \centering
    \begin{tabular}{lrr}
        \toprule
        Model & $\mmd$ & p-value \\
        \midrule
        Base model & -0.578 & 0.628 \\
        No $\Delta t_i$ & 2.991 & 0.026 \\
        No $r_{t_{i-1}}$ & 12.124 & 0.0 \\
        No $r_{t_{i-1}}$ and no $\Delta t_i$ & 620.229 & 0.0 \\
        \bottomrule
    \end{tabular}
    \caption{Ablation study on the importance of the inputs to the LSTM with 10,000 permutations}
    \label{tab:ablation_mmd}
\end{table}

Next, we investigate the effect of conditioning the LSTM cell state and hidden state with historical data as described in Section \ref{sec:conditioning}.
We train two sets of models that differ only by whether the initial cell state and hidden state of the LSTM is conditioned with historical data or initialised to zero during training.
As the random seeds were set to the same value, the specific sequences sampled to train the models and the initialisation of the weights are the same.
For each model, we generate 100 samples with and without conditioning the trained model with historical data i.e. there are 4 sets of samples.
All sequences are generated with the same noise input and the tests are conducted with the same random seed.
Table \ref{tab:conditioning_mmd} shows the $\mmd$ and p-value of the two sample tests.
By comparing the $\mmd$ scores, which increase without the procedure and by a decreasing p-value, we observe that conditioning the LSTM with historical data has a significant postive effect on the performance.

\begin{table} [!ht]
    \centering
    \begin{tabular}{lrr}
        \toprule
        Model & $\mmd$ & p-value \\
        \midrule
        Conditionally trained and generated & -0.578 & 0.628 \\
        Conditionally trained but non-conditionally generated & 46.076 & 0.0 \\
        Non-conditionally trained but conditionally generated & 2.296 & 0.052 \\
        Non-conditionally trained and generated & 12.587 & 0.0 \\
        \bottomrule
    \end{tabular}
    \caption{Effect of conditioning the LSTM with historical data on $\mmd$ with 10,000 permutations}
    \label{tab:conditioning_mmd}
\end{table}

\subsection{Sensitivity Analysis} \label{sec:sensitivity}

We perform a sensitivity analysis by perturbing the inputs to the generative model and observing the effect on the output.
From Section \ref{sec:hyperparameters}, we know that there is a dominant noise dimension that has an outsized impact on the generated sequences.
Therefore, we pick this particular dominant dimension to perturb when it comes to the noise input with all other dimensions set to zero.

We studied sample sizes of 593,000 and 588,000, then rounded the extreme values to set the range of perturbation to $[-7.5,7.5]$ and $[-0.15,0.15]$ for the noise and returns respectively.
Note that there were extreme outliers in the samples as 99.9\% of the samples were within the range of $[-2.88,2.91]$ and $[-0.0543,0.0436]$ for the noise and returns respectively.
Although we know that $\Delta t_i$ has an impact on the performance of the model from Section \ref{sec:ablation}, the sensitivity is extremely low hence is omitted.

We chose three periods of the S\&P 500 data with different market environments to condition the LSTM.
The first period from 23 Mar to 4 Jun 2007 was upward trending, the second period from 18 Aug to 27 Oct 2008 was downward trending and the last period from 31 Mar to 10 Jun 2015 was sideways trending.
In addition, we also experiment with no conditioning of the LSTM cell state and hidden state i.e. initialised to zero.

Figure \ref{fig:noise_prev_output} shows the sensitivity of the output (log return) with respect to the noise and the previous return.
For clarity, we only show the figure for the period from 18 Aug to 27 Oct 2008\footnote{\label{note:fig}Omitted figures can be found at \url{https://github.com/Anonymous1701/mmd-repo/tree/main/figures}}.
More importantly, it is apparent that the noise has a much larger impact on the output than the previous return.
As it is difficult to see the relationship between the output and the previous return from Figure \ref{fig:noise_prev_output}, we plot out the cross sections at fixed noise levels in Figure \ref{fig:sen_autocorr}.
Interestingly, the sensitivity to the previous return is larger at higher absolute noise levels with more positive returns resulting in larger absolute returns and vice versa.

\begin{figure}
    \centering
    \subfigure[Noise vs previous return vs output]{
        \includegraphics[scale=0.35]{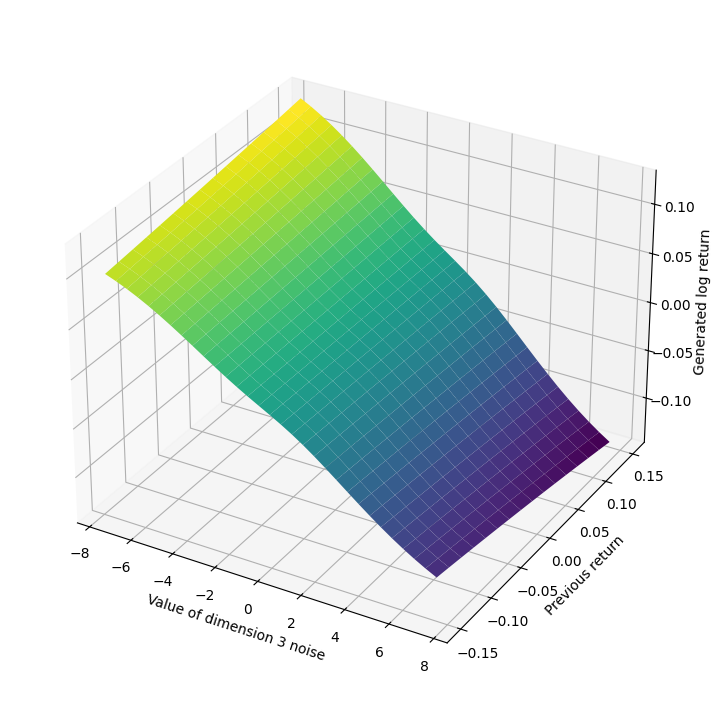}
        \label{fig:noise_prev_output}
    }
    \subfigure[Previous return vs output]{
        \includegraphics[scale=0.45]{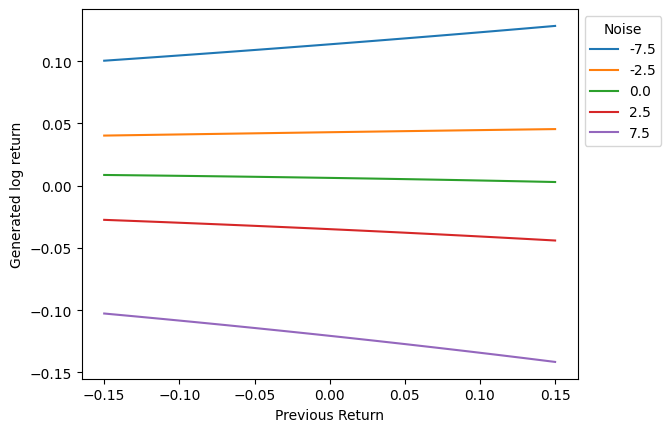}
        \label{fig:sen_autocorr}
    }
    \caption{Left plot shows sensitivity of output with respect to the noise and the previous return. Right plot shows the sensitivity of the output with respect to the previous return with the noise at a fixed value for each line.}
\end{figure}

Next, we investigate the effect of noise after different conditioning periods.
Figure \ref{fig:sen_noise_hist} shows the sensitivity of the output with respect to the noise, where the previous return is from the last data point in the historical period.
The change from a bullish to sideways to bearish market environment appears to cause the plot to rotate clockwise slightly which means that volatility will be higher following bearish periods than bullish periods.

As the noise is essentially Gaussian with zero mean, there is equal probability of the noise being positive or negative.
Therefore, if the y-intercept of the plot is positive (negative), it indicates that there is on balance a higher probability of seeing a positive (negative) return.
For all plots generated with conditioning in Figure \ref{fig:sen_noise_hist}, the y-intercept is slightly positive.
The plot generated without conditioning deviates from the other plots significantly which would result in a very different distribution of returns, consistent with the results in Section \ref{sec:ablation}.

Figure \ref{fig:asymmetry} shows the sum of the two log returns generated when the noise is at the positive and negative value of the x-axis.
In other words it shows the skew of the returns at different noise values, where a positive (negative) y-axis value indicates a positive (negative) skew.
For all 3 plots, the x-intercept is beyond the 99.9\% range of values seen in the noise samples, which combined with the earlier observation that there is a higher chance of seeing a positive return, means that for the vast majority of the time, there is a higher chance of a positive drift in the index.
This is consistent with the stylised fact that the mean of the returns is typically positive in the long run.

\begin{figure}
    \centering
    \subfigure[Noise vs output]{
        \includegraphics[scale=0.45]{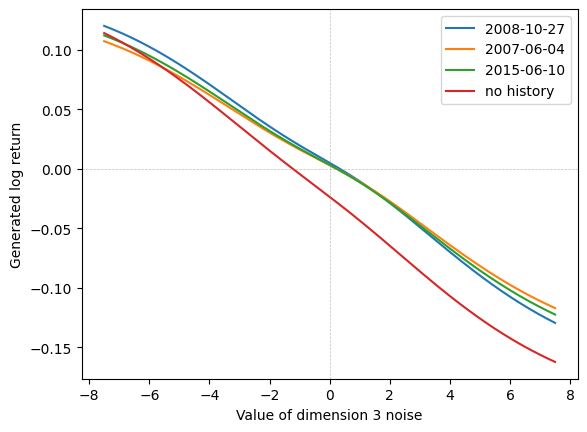}
        \label{fig:sen_noise_hist}
    }
    \subfigure[Asymmetry of the noise vs output]{
        \includegraphics[scale=0.45]{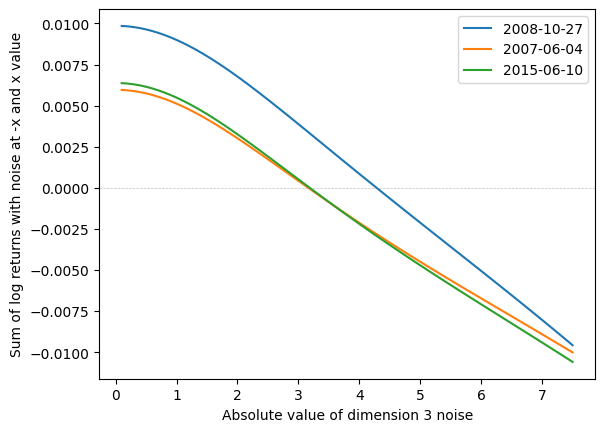}
        \label{fig:asymmetry}
    }
    \caption{Left plot shows sensitivity of output with respect to the noise. Right plot shows the sum of the two log returns generated when the dominant noise value is at the positive and negative value of the x-axis. Each line has a different historical period used for conditioning with the end date of the period indicated in the legend.}
\end{figure}

\subsection{Stylised Facts} \label{sec:stylised_facts}

In this section, we demonstrate that the generative model is able to produce sequences that exhibit most of the behaviours in line with well known stylised facts (see e.g. \cite[Section 2]{cont2001empirical}) of financial time series.
Recall that the training procedure involves sampling segments of the historical S\&P 500 data, which is a sequence of 250 prices.
We compare the statistics of the generated sequences, which have the same length as the segments, with the statistics of the segments and full sequence i.e. from 1995 to 2018.
We generate sequences with the same procedure as in training, with $k=50,n=299$ in the notation of equation \eqref{eq:trunc}, and only use the truncated sequences, $\x_{\Theta,\text{trun}}$ in Algorithm \ref{alg:training}, for the analysis.

{
As a comparison, we train a signature kernel based neural SDE (NSDE), COT-GAN and Conditional Sig-WGAN (C-Sig-WGAN) generative models introduced in \cite{issa2023non}, \cite{xu2020cot}, and \cite{liao2024sig}, respectively using the same data.
We also use the same MA model to generate the noise input to the COT-GAN model, which has the same architecture including the number of noise dimensions.
To ensure a fair comparison, we also use the same noise sequence for both our model and the COT-GAN model to generate the samples of the same length.
The Conditional Sig-WGAN model is set up with the original optimised specification in \cite{liao2024sig} except with the signature truncation order increased to 5 which was the limit of the GPU memory available for the algorithm to run.
For the NSDE model, we use the same hyperparameters as in \cite{issa2023non} for the unconditional setting except for the data related hyperparameters which are adjusted to match our use case.
We also tried varying the $\sigma$ parameter of the RBF kernel in the NSDE model but found that it did not have a significant impact on the results hence we stuck to the original value of 1.0.
We could not use the conditional setting of the NSDE model as the hyperparameters of the orignal model was unable to produce any meaningful results.
We modify the code from the authors' repository available at \url{https://github.com/issaz/sigker-nsdes} for NSDE, \url{https://github.com/tianlinxu312/cot-gan} for COT-GAN and \url{https://github.com/SigCGANs/Conditional-Sig-Wasserstein-GANs} for Conditional Sig-WGAN for our specific use case.
}

\subsubsection{Moments of Returns}

The distribution of asset returns is known to have fatter tails and taller peaks than the Gaussian distribution (see e.g. \cite{cont2001empirical}).
It also exhibits a negative skew.
In the long run, the mean of the returns is typically positive.
We study the first 4 moments of the return distribution for these characteristics\footnote{Traditionally in finance, the average return and volatility of returns are annualised based on business days.
Therefore, all returns are on equal footing and calendar days are not taken into account.
We follow this approach when calculating the statistics. All returns are also calculated as log returns instead of simple returns.
Refer to Appendix \ref{app:returns_stats} for the formulas.}.
To this end, we generate 10,000 sequences using our MMD based generative model and calculate the moments of the log returns.
We do the same by sampling the S\&P 500 segments.
Figure \ref{fig:returns_stats} shows the histogram of the moments across these samples.
The blue bars are the moments of the S\&P 500 segments while the orange bars are the moments of the generated sequences and the red lines are the moments of the full S\&P 500 sequence.

It is clear that there can be significant variations in the statistics of the segments compared to the full sequence of S\&P 500.
Since the training of the generative model is based on the segments, it is not surprising and also desirable that the statistics of the generated sequences follow closer to the distribution of the segments.

\begin{figure} [!ht]
    \centering
    \subfigure[Mean of log returns]{
        \includegraphics[scale=0.5]{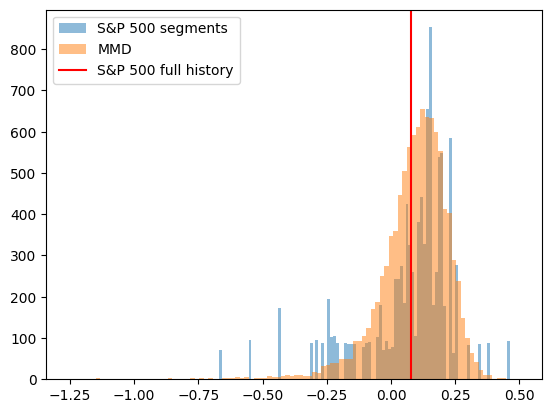}
        \label{fig:mean}
    }
    \subfigure[Standard deviation of log returns]{
        \includegraphics[scale=0.5]{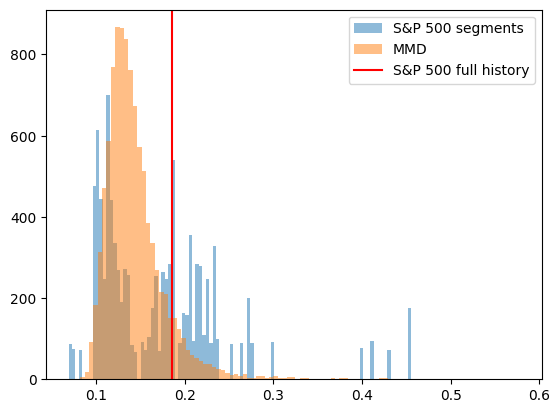}
        \label{fig:std}
    }
    \subfigure[Skew of log returns]{
        \includegraphics[scale=0.5]{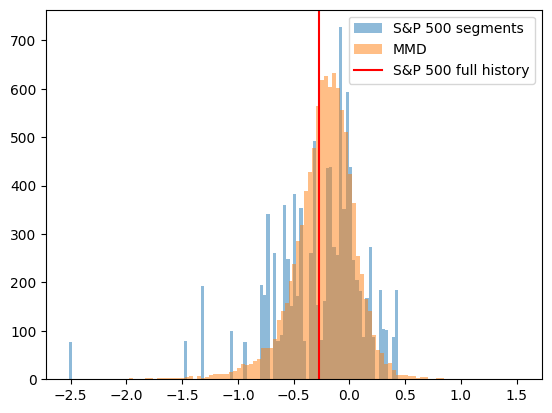}
        \label{fig:skew}
    }
    \subfigure[Kurtosis of log returns]{
        \includegraphics[scale=0.5]{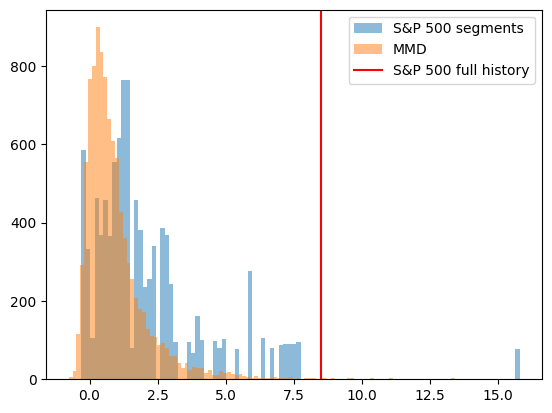}
        \label{fig:kurt}
    }
    \caption{Statistics of log return. Top left plot is the mean of the annualised daily log returns. Top right plot is the standard deviation of the annualised daily log returns. Bottom left plot is the skew of the daily log returns. Bottom right plot is the kurtosis of the daily log returns. The blue bars are the moments of the S\&P 500 segments, the orange bars are the moments of the generated sequences and the red lines are the moments of the full S\&P 500 sequence.}
    \label{fig:returns_stats}
\end{figure}

In addition, we also include the statistics of the moments for both our MMD based generator, the NSDE generator, the COT-GAN generator and the C-Sig-WGAN generator in Table \ref{tab:returns_stats}.
The table shows the mean and standard deviation of the distribution of the moments across 10,000 samples.
In general, the models more or less capture the characteristics of the S\&P 500.
{
One deviation is in the samples from the NSDE and C-Sig-WGAN models which on average produces a slightly positive skew.
The NSDE model also does not produce fatter tails which is evident in the kurtosis of the distribution.
}
The COT-GAN samples have a much wider distribution of the first moment compared to the other models and the S\&P 500 segments.
This is also evident when we look at the distribution of the end points of the generated segments.

\begin{table} [!ht]
    \centering
    {
    \begin{tabular}{lllll}
        \toprule
        Return Type & Ann. returns & Volatility & Skew & Kurtosis \\
        \midrule
        S\&P 500 full history & 0.077 / - & 0.185 / - & -0.271 / - & 8.495 / - \\
        S\&P 500 segments & 0.069 / 0.181 & 0.172 / 0.075 & -0.270 / 0.415 & 2.129 / 2.267 \\
        MMD & 0.091 / 0.135 & 0.143 / 0.033 & -0.222 / 0.278 & 0.955 / 1.206 \\
        NSDE & 0.069 / 0.184 & 0.149 / 0.026 & 0.046 / 0.168 & 0.112 / 0.385 \\
        COT-GAN & 0.086 / 0.340 & 0.240 / 0.045 & -0.130 / 0.296 & 0.663 / 0.977 \\
        C-Sig-WGAN & 0.087 / 0.155 & 0.175 / 0.038 & 0.013 / 0.665 & 5.050 / 4.544 \\
        \bottomrule
    \end{tabular}
    }
    \caption{Statistics of the first 4 moments of the log returns. The first value is the mean and the second value is the standard deviation of the distribution of the moments across 10,000 samples.}
    \label{tab:returns_stats}
\end{table}

Figure \ref{fig:end_hist} shows the density plot of the generated sequences' end points with the initial points normalised to 1.
Note that we truncated the x-axis for better visualisation as 35 of the COT-GAN samples had end points scattered between 2.5 and 5.6.
The end points of the MMD generated sequences are more aligned with the S\&P 500 segments than the COT-GAN generated sequences.
The presence of more extreme returns explains the higher 2nd moment and lower 4th moment of the COT-GAN generated sequences due to a fatter middle.
The extreme values of the COT-GAN generated sequences are overly pronounced and are not a realistic reflection of the training data.
For the {NSDE} and C-Sig-WGAN generated sequences, it is closer to the training data but the MMD generated sequences have a better overall fit.

\begin{figure} [!ht]
    \centering
    \subfigure[MMD]{
        \includegraphics[scale=0.5]{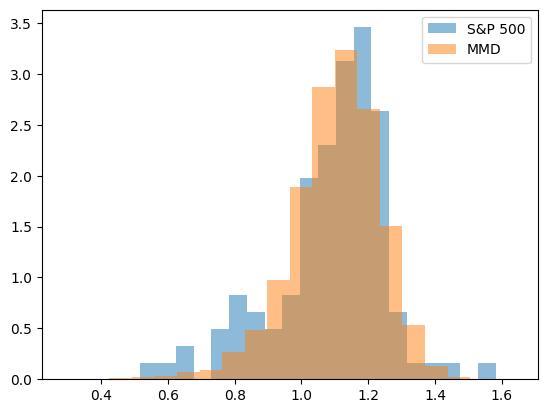}
    }
    \subfigure[NSDE]{
        \includegraphics[scale=0.5]{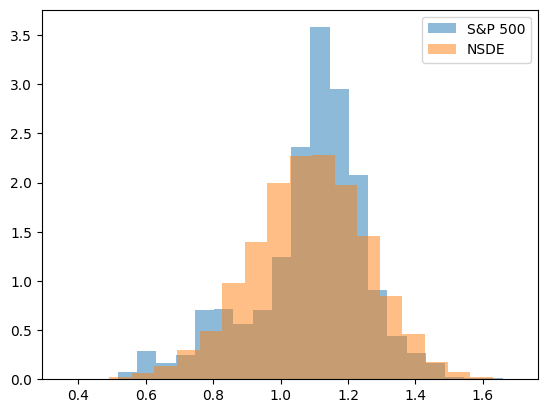}
    }
    \subfigure[COT-GAN]{
        \includegraphics[scale=0.5]{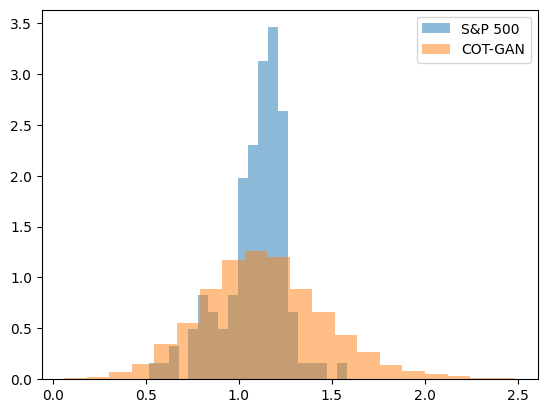}
    }
    \subfigure[C-Sig-WGAN]{
        \includegraphics[scale=0.5]{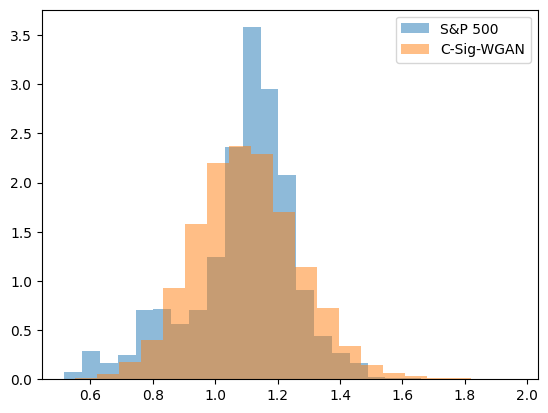}
    }
    \caption{Density plot of the end points of the generated sequences. Initial values are normalised to start from 1.0 for all sequences. The x-axis is truncated for COT-GAN (due to more extreme vales). The S\&P 500 histogram is constructed based on the training data construction method which are the same for MMD, NSDE and COT-GAN but slightly different for C-Sig-WGAN which has shorter lengths.}
    \label{fig:end_hist}
\end{figure}

\subsubsection{Autocorrelation in Returns}

Equity market returns are known to lack linear autocorrelation (see e.g. \cite{cont2001empirical}) but exhibit certain types of nonlinear autocorrelation.
Figure \ref{fig:returns_acf} shows the linear and non-linear autocorrelation of the log returns.
The blue and orange lines are the mean of the autocorrelation coefficient across the generated and sampled sequences respectively while the shaded region shows the range of the mean absolute deviation.
The red line is the autocorrelation coefficient of the full S\&P 500 sequence.

In Figure~\ref{fig:returns_acf}, the statistics are consistent across the 3 types of sequences where it shows a lack of linear autocorrelation.
We omit the comparison for linear autocorrelation of the returns as the NSDE, COT-GAN and C-Sig-WGAN generators produce similar results.\footnote{See footnote \ref{note:fig}}.

\begin{figure} [!ht]
    \centering
    \includegraphics[scale=0.6]{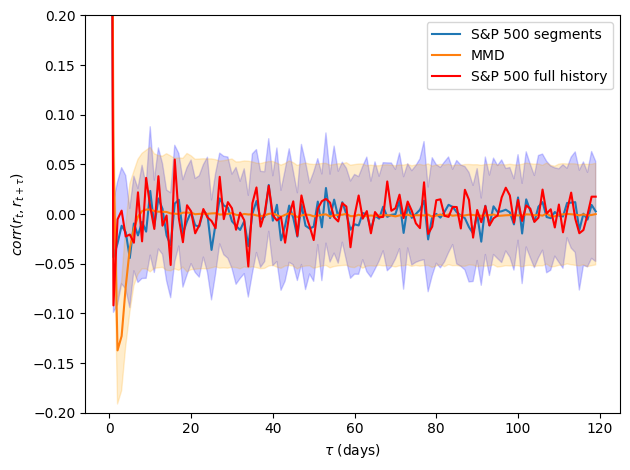}
    \caption[Linear autocorrelation of log returns]{Linear autocorrelations of log returns. The blue and orange lines are the mean of the coefficient across generated sequences, while the shaded region covers the mean absolute deviation. The red line is the coefficient of the full S\&P 500 sequence.}
    \label{fig:returns_acf}
\end{figure}

However, returns are certainly not independent due to the phenomenon of volatility clustering.
There are different manifestations of volatility clustering, and one of them is the autocorrelation of the squared returns.
{
Figure~\ref{fig:r_sq_comp} shows the autocorrelation of the squared returns for all generated sequences.
While both the S\&P 500 segments and full sequence exhibit positive and slow decaying autocorrelation with respect to an increasing time lag, the full sequence shows a more pronounced effect.
The MMD and COT-GAN generated sequences are a closer match to the S\&P 500 segments.
The C-Sig-WGAN generated sequences exhibit a much faster and smoother decay in the autocorrelation of the squared returns.
The NSDE generated sequences do not capture the autocorrelation of the squared returns well, with an almost flat autocorrelation curve.
}

\begin{figure} [!ht]
    \centering
    \subfigure[MMD]{
        \includegraphics[scale=0.44]{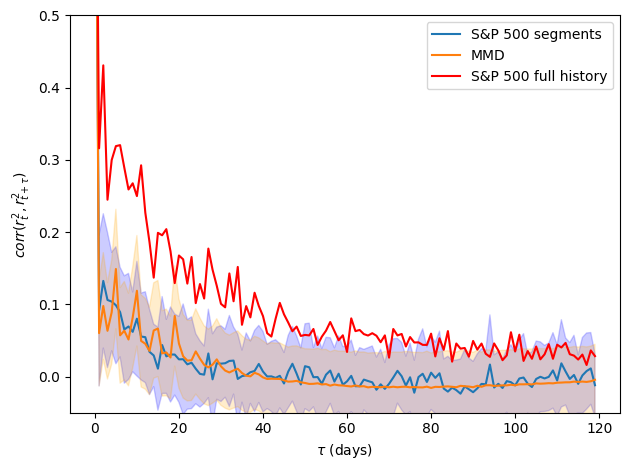}
        \label{fig:mmd_r_sq}
    }
    \subfigure[NSDE]{
        \includegraphics[scale=0.44]{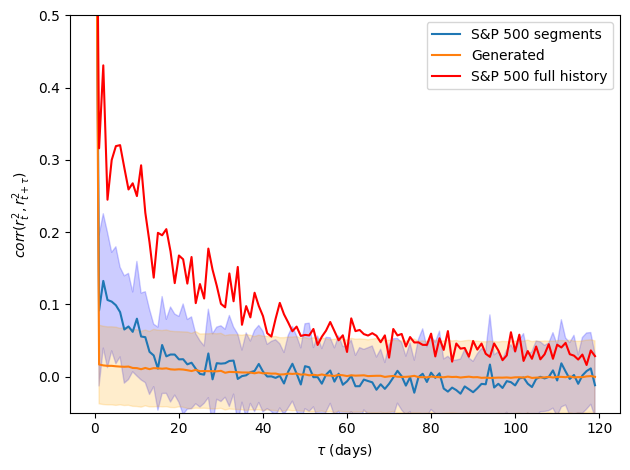}
        \label{fig:nsde_r_sq}
    }
    \subfigure[COT-GAN]{
        \includegraphics[scale=0.44]{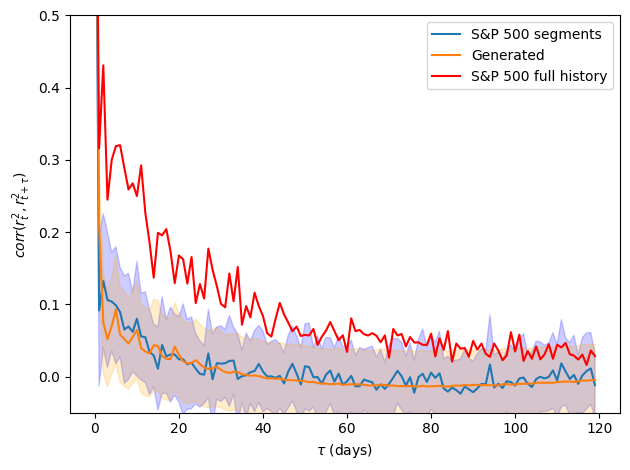}
        \label{fig:cotgan_r_sq}
    }
    \subfigure[C-Sig-WGAN]{
        \includegraphics[scale=0.44]{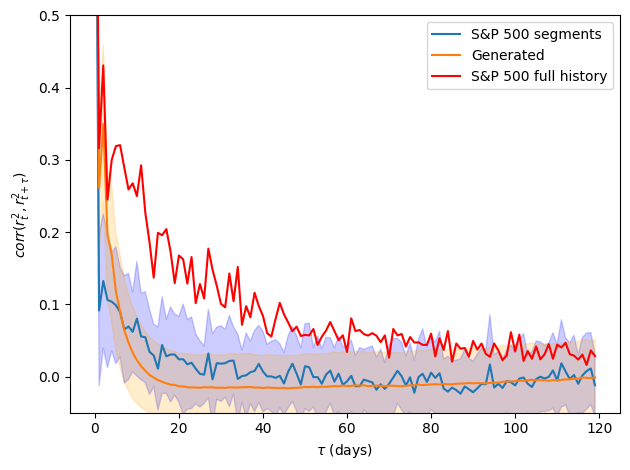}
        \label{fig:csigwgan_r_sq}
    }
    \caption[Autocorrelation of squared returns for all generated sequences]{Autocorrelation of squared returns for all generated sequences. Lines and shaded regions are the same as in Figure \ref{fig:returns_acf}.}
    \label{fig:r_sq_comp}
\end{figure}

\subsubsection{Gain/loss Asymmetry and Leverage Effect}

The negative skew in the observed returns can also be viewed through different lenses.
The gain/loss asymmetry is the phenomenon that there are more large negative returns observed than large positive returns (see e.g.\cite{cont2001empirical}).
Figure \ref{fig:gain_loss} shows the expected ratio of positive returns (y-axis) given that the absolute value of the return exceeds a certain threshold (x-axis).
For the threshold, we only use values up to 0.0237 which corresponds to the 95th percentile of the returns found in S\&P 500 since there are limited samples beyond this point.
On this measure, the MMD generated sequences are more aligned to the S\&P 500 segments displaying a similar level of asymmetry.
{
However, both the NSDE and COT-GAN generated sequences have a return distribution that is more symmetric.
}
For the C-Sig-WGAN generated sequences, the asymmetry more or less follows the same pattern as the S\&P 500 but is much more pronounced until it gets closer to the 95th percentile.

The leverage effect is the observation that the volatility of the returns is negatively correlated with the returns (see e.g.\cite{cont2001empirical}).
We can quantify this by calculating the correlation between the returns and the squared returns in the subsequent days as shown in Figure \ref{fig:leverage_effect}.
On this measure, we see that the MMD generated sequences only show a weak negative correlation in the immediate subsequent day before it quickly decays towards zero.
This is in contrast to the S\&P 500 segments and full sequence which shows a more pronounced negative correlation that is slower to decay.
The COT-GAN generated sequences show a similar pattern to the MMD generated sequences while the C-Sig-WGAN exhibits a pronounced negative correlation in the immediate subsequent day before it quickly decays towards zero.
{
The NSDE generated sequences has the opposite effect with a positive correlation in the immediate subsequent day before it decays towards zero.
}

\begin{figure} [!ht]
    \centering
    \includegraphics[scale=0.5]{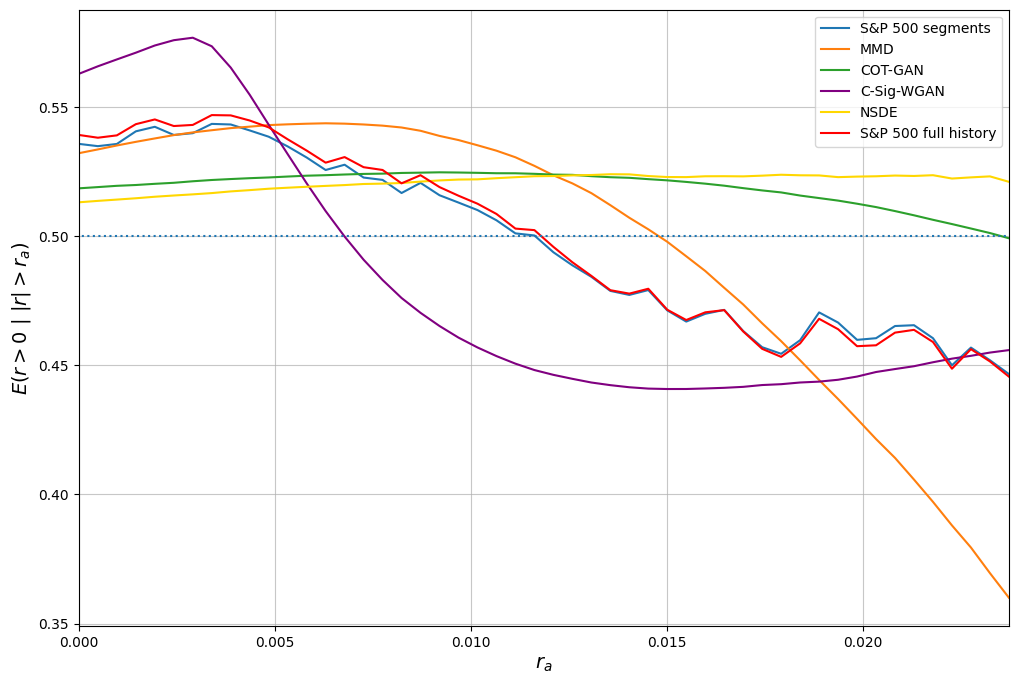}
    \caption[Gain/loss asymmetry]{The gain/loss asymmetry of the daily log returns. The y-axis is the expected ratio of positive returns given that the return exceeds a certain threshold on the x-axis.}
    \label{fig:gain_loss}
\end{figure}

\begin{figure} [!ht]
    \centering
    \subfigure[MMD]{
        \includegraphics[scale=0.42]{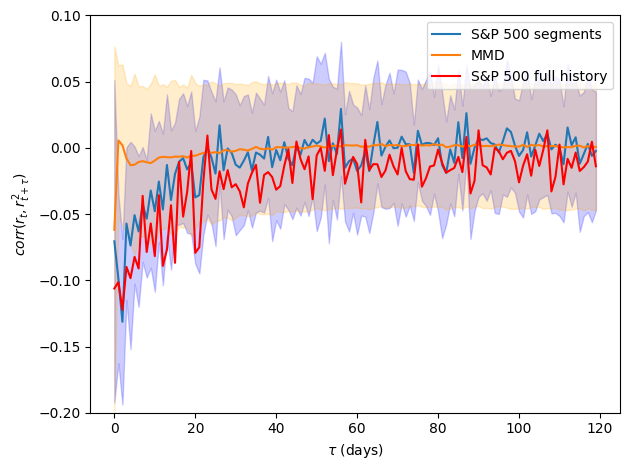}
        \label{fig:leverage_effect}
    }
    \subfigure[NSDE]{
        \includegraphics[scale=0.42]{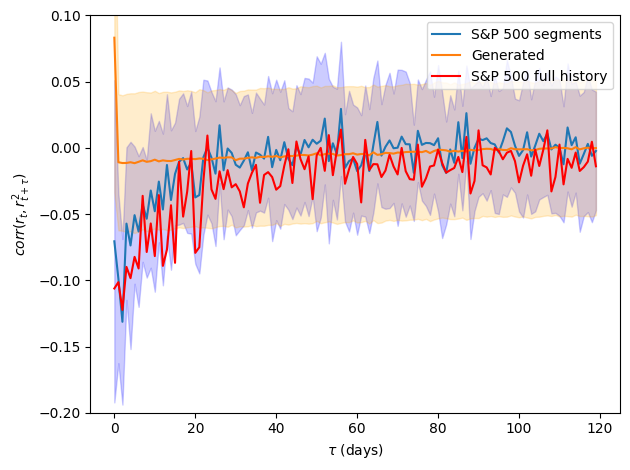}
        \label{fig:mmd_nsde_lev}
    }
    \subfigure[COT-GAN]{
        \includegraphics[scale=0.42]{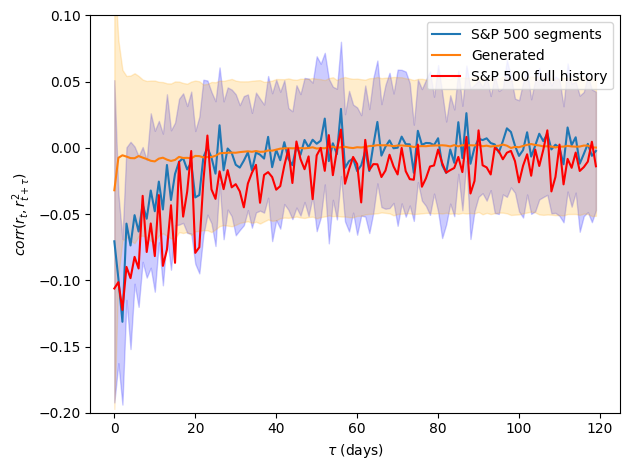}
        \label{fig:mmd_cotgan_lev}
    }
    \subfigure[C-Sig-WGAN]{
        \includegraphics[scale=0.42]{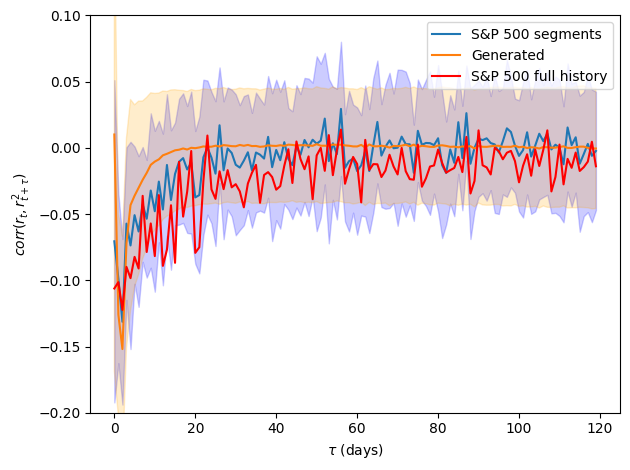}
        \label{fig:mmd_csigwgan_lev}
    }
    \caption[Leverage effect]{The plots show the correlation between the daily log returns and the squared returns of the subsequent $\tau$ days with lines and shaded regions being the same as in Figure \ref{fig:returns_acf}.}
\end{figure}

\subsubsection{Distance Based Comparison} \label{sec:distance_based_comparison}

{
Finally, we compare the generated sequences using two distance-based measures for time series distributions.
First, we use $\mmd$ to perform the permutation tests as in Section~\ref{sec:noise}, but using the \emph{untruncated} signature based on \cite{salvi2021signature} for the signature kernel.
This particular signature kernel is computed by numerically solving a related PDE and is computationally expensive.
We choose a resolution for the grid such that the results do not change with further refinement.
Second, we use the sliced adapted Wasserstein distance (SAW), also used in \cite{acciaio2024time}, which is computationally less expensive than the adapted Wasserstein distance.
}

{
There are 114 S\&P 500 segments available and we extract corresponding sequences with the same time period to calculate both distances.
Table~\ref{tab:distance_comparison} shows the results of the distance-based comparison.
The MMD generated sequences perform the best for both distance measures, with the C-Sig-WGAN generated sequences coming second on both measures.
Both the NSDE and COT-GAN generated sequences perform poorly on both measures.
}

\begin{table}
    \centering
    {
    \begin{tabular}{lrrrr}
        \toprule
        Model & $\mmd$ & p-value & SAW mean & SAW std \\
        \midrule
        MMD & \textbf{0.479} & \textbf{0.11} & \textbf{0.133} & 0.010 \\
        NSDE & 3.561 & 0.0 & 0.153 & 0.010 \\
        COT-GAN & 8.967 & 0.0 & 0.285 & 0.020 \\
        C-Sig-WGAN & 0.762 & 0.05 & 0.137 & \textbf{0.007} \\
        \bottomrule
    \end{tabular}
    }
    \caption[Distance based comparison of the generated sequences with the S\&P 500 segments]{Distance based comparison of the generated sequences with the S\&P 500 segments using $\mmd$ based permutation test with a an \emph{untruncated} signature kernel and sliced adpated Wasserstein distance (SAW).}
    \label{tab:distance_comparison}
\end{table}

\subsection{Portfolio Management using Reinforcement Learning} \label{sec:portfolio_management}

One of the potential applications of a financial time series generative model is portfolio management.
As the amount of data available is relatively small for machine learning standards, a generative model that captures the essential characteristics of the financial time series can be used to generate additional data for training an agent.
To do this, we train a reinforcement learning (RL) agent to manage between cash and the S\&P 500 index based training data from generated sequences.
The environment is set up with zero interest rates, a constraint of zero leverage\footnote{As the S\&P 500 is generally upward trending, allowing leverage made it easier for the agent to outperform the index based on initial experiments, hence this leverage constraint was added to increase the challenge.}, and 5 basis points of proportional transaction cost on the notional amount traded.
The agent is allowed to short the S\&P 500 index up to 100\% of the portfolio value without considering borrowing costs.
Such a setting is not realistic but allows us to focus on whether the agent is able to learn a strategy that exploits the characteristics captured by the generative model.

\begin{figure}
    \centering
    \includegraphics[scale=0.58]{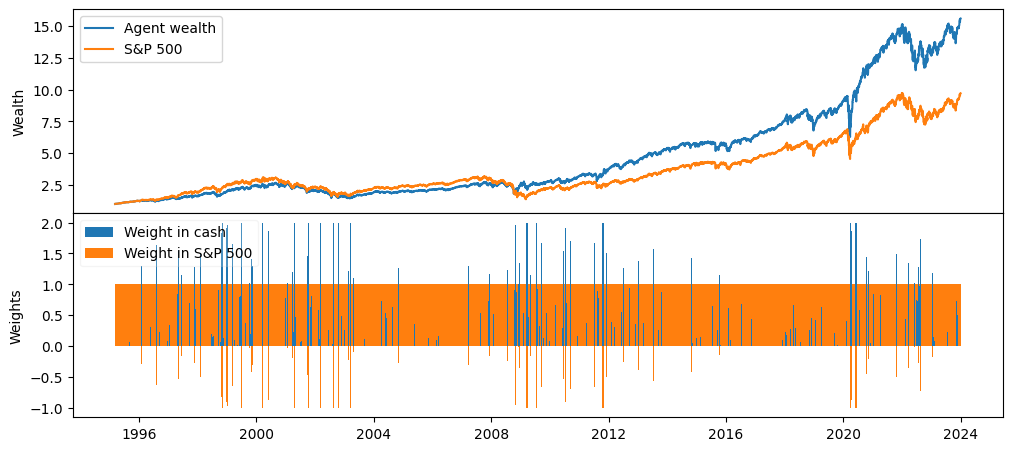}
    \caption{Performance of RL agent on S\&P 500 index on the full period of data. The baseline is 100\% long on the index. Wealth is normalised to 1 at the start.}
    \label{fig:full_perf}
\end{figure}

The agent is trained using the Proximal Policy Optimisation (PPO) algorithm introduced in \cite{schulman2017proximal} with settings based on best practices from previous work in \cite{lu2023evaluation}.
We adopt the same approach to generate the data as for the training of the generative model.
In other words, we do not generate a fixed dataset and consistently update the sequences used for training the RL agent.
Note that although the initial historical portion of the sequence (used to condition the LSTM cells' state and hidden state) is used as the initial observation state for the agent, all returns on the portfolio that are based on the agent's actions during training are calculated on generated sequences only.

We generate longer sequences \footnote{Using the sequence length in training with $n=299$ resulted in the agent learning a constant 100\% long the index strategy hence the switch to a longer length of $n=1299$.} and observe that the agent learns active trading strategies with sporadic periods where she will short or reduce the position in the S\&P 500 index.
Figure \ref{fig:full_perf} shows the performance of the agent on the historical data of S\&P 500 index including the test period from 19 Sep 2018 to 28 Dec 2023 while Figure \ref{fig:test_perf} zooms in on the test period of the data.
Table \ref{tab:perf} shows the annualised returns, volatility, Sharpe ratio, and maximum drawdown of the agent and the S\&P 500 during the full period and the test period.
The agent is able to outperform the S\&P 500 both over the full period and the test period with higher returns, lower volatility and hence higher Sharpe ratio.
The outperformance occurs as the agent learns to short the S\&P 500 index at the correct times, that is, anticipates drops in the index.
Although this is still a relatively small sample that makes it difficult to draw strong conclusions, the results indicate potential for this approach.

\begin{table}
    \centering
    \begin{tabular}{lrrrr}
        \toprule
        & Ann. return & Volatility & Sharpe ratio \tablefootnote{Calculated as the annualised log return divided by volatility with interest rate assumed to be zero.} & Max drawdown \tablefootnote{Calculated based on log returns.} \\
        \midrule
        S\&P 500 (full period) & 0.079 & 0.191 & 0.413 & -0.839 \\
        RL agent (full period) & 0.095 & 0.180 & 0.532 & -0.621 \\
        S\&P 500 (test period) & 0.094 & 0.215 & 0.440 & -0.414 \\
        RL agent (test period) & 0.118 & 0.205 & 0.573 & -0.414 \\
        \bottomrule
    \end{tabular}
    \caption{Performance of RL agent on S\&P 500 index. The full period is from 1 Jan 1995 to 28 Dec 2023 while the test period is from 19 Sep 2018 to 28 Dec 2023. The annualised returns and volatility are calculated as in Table \ref{tab:returns_stats}.}
    \label{tab:perf}
\end{table}

\begin{figure}
    \centering
    \includegraphics[scale=0.58]{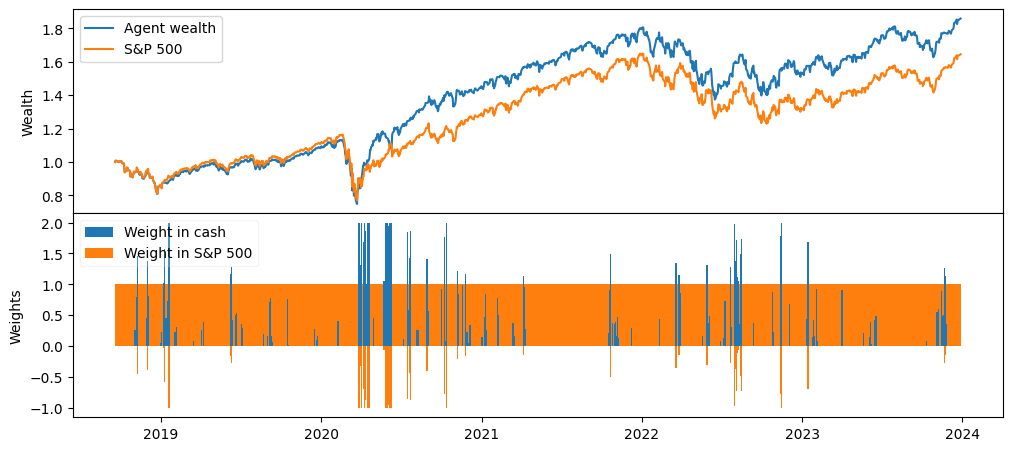}
    \caption{Performance of RL agent during the out of sample period. Normalisation and baseline are the same as in Figure \ref{fig:full_perf}.}
    \label{fig:test_perf}
\end{figure}

It is noteworthy that the agent's performance continued to improve (measured on training data) up to at least 10m steps of training, which corresponds to almost 40,000 years of data.
Although further fine-tuning of training hyperparameters could reduce the amount of data required to achieve the same performance, it is unlikely that the agent can achieve the same performance using only the historical data.
This emphasises the importance having a generative model to provide additional data, with sufficient fidelity to the distribution of the historical data, for training the RL agent.

\subsection{Adding Robustness}

We propose a method to improve the robustness of using the generative model to train the RL agent by altering the noise distribution.
Specifically, we selectively only fit the MA model to sequences that correspond to market downturns, which we define as periods where the market has declined by 30\% or more.
As the data from periods of market downturns are even more limited, it is not possible to train a generative model directly on this data.
Our method is a way to inject characteristics of market downturns into the generated sequences without using large amounts of data.
This allows us to adapt to the generator to distributions changes which likely occur during market downturns.

At the same time, to ensure that we are not too adrift from the true distribution, we perform permutation tests.
The test is performed on selected S\&P 500 segments from market downturns in addition to the complete data set against the generated sequences using altered noise, henceforth called robust noise.
We selected the dot-com bubble burst and the 2008 Global Financial Crisis (GFC) as the market downturns to focus on.
In addition, we use robust noise generated using the average of the parameters of the MA model trained on the dot-com bubble burst and the 2008 GFC periods.
This method of averaging the parameters in a model has been shown to improve the performance of a model in a deep learning context (see e.g. \cite{wortsman2022model}, \cite{rame2022diverse}), hence we experiment with this approach to see if it can improve the robustness of the generative model.
Table \ref{tab:robust_noise_mmd} shows the results of the permutation test.
It is unsurprising that the robust noise generates sequences that do not match the overall distribution of the S\&P 500 segments.
However, when we focus on market downturns, robust noise generates sequences that are closer to the S\&P 500 segments than the base generator which uses noise trained on the full dataset.

\begin{table}
    \centering
    \begin{tabular}{lll}
        \toprule
        & p-value (full dataset) & p-value (market downturns) \\
        \midrule
        Base generator & 0.218 & 0.021 \\
        Dot-com bubble burst generator & 0.000 & 0.068 \\
        2008 GFC generator & 0.101 & 0.200 \\
        Averaged generator & 0.009 & 0.194 \\
        \bottomrule
   \end{tabular}
   \caption{Results of permutation test on the S\&P 500 segments from the full dataset and market downturns against the generated sequences using the robust noise. The base generator has a noise generator trained on the entire dataset while the other generators have noise generators trained on selected periods that encompass the named market downturns with the average generator using the average of the MA model parameters.}
    \label{tab:robust_noise_mmd}
\end{table}

\begin{table}
    \centering
    \begin{tabular}{lrrrrrrrrr}
        \toprule
        & Ann. return & Volatility & Sharpe ratio & Max drawdown \\
        \midrule
        \multicolumn{4}{l}{Dot-com bubble burst from 24 Mar 2000 to 9 Oct 2002} \\
        \midrule
        S\&P 500 & -0.267 & 0.229 & -1.167 & -0.676 \\
        Base generator & -0.196 & 0.214 & -0.918 & -0.601 \\
        Dot-com bubble burst generator & \textbf{0.171} & \textbf{0.172} & \textbf{0.994} & \textbf{-0.149} \\
        2008 GFC generator & -0.061 & 0.187 & -0.328 & -0.305 \\
        Averaged generator & 0.138 & 0.179 & 0.769 & -0.164 \\
        \midrule
        \multicolumn{4}{l}{2008 GFC from 9 Oct 2007 to 20 Nov 2008} \\
        \midrule
        S\&P 500 & -0.650 & 0.362 & -1.803 & -0.732 \\
        Base generator & -0.346 & 0.334 & -1.038 & -0.426 \\
        Dot-com bubble burst generator & \textbf{0.307} & \textbf{0.283} & \textbf{1.089} & \textbf{-0.219} \\
        2008 GFC generator & -0.004 & 0.284 & -0.015 & -0.293 \\
        Averaged generator & 0.273 & 0.294 & 0.933 & -0.262 \\
        \midrule
        \multicolumn{4}{l}{COVID-19 pandemic from 19 Feb 2020 to 23 Mar 2020} \\
        \midrule
        S\&P 500 & -4.351 & 0.792 & -5.732 & -0.414 \\
        Base generator & -4.351 & 0.792 & -5.732 & -0.414 \\
        Dot-com bubble burst generator & \textbf{2.362} & \textbf{0.509} & \textbf{4.840} & \textbf{-0.057} \\
        2008 GFC generator & 0.289 & 0.704 & 0.429 & -0.152 \\
        Averaged generator & 0.835 & 0.714 & 1.221 & -0.158 \\
            \midrule
        \multicolumn{4}{l}{Full dataset from 1 Jan 1995 to 28 Dec 2023} \\
        \midrule
        S\&P 500 & 0.079 & 0.191 & 0.413 & -0.839 \\
        Base generator & \textbf{0.095} & 0.180 & \textbf{0.532} & \textbf{-0.621} \\
        Dot-com bubble burst generator & -0.008 & \textbf{0.151} & -0.053 & -0.962 \\
        2008 GFC generator & 0.027 & 0.156 & 0.174 & -0.654 \\
        Averaged generator & -0.009 & 0.155 & -0.061 & -1.058 \\
        \bottomrule
    \end{tabular}
    \caption{Performance of RL agent trained with different noise distributions on S\&P 500 index during different market downturn periods. The annualised returns, volatility, Sharpe ratio and maximum drawdown are calculated as in Table \ref{tab:perf}. The evaluation is done on the S\&P 500 index during the named market downturn periods indicated.}
    \label{tab:robust_noise_perf}
\end{table}

We repeat the training procedure from Section \ref{sec:portfolio_management} using the robust noise generated sequences.
In particular, we focus on the performance of the RL agent during the market downturns.
Table \ref{tab:robust_noise_perf} shows the performance of the RL agent trained with different noise distributions on the S\&P 500 index during different market downturn periods.
Both RL agents trained using generators with robust noise outperform the base generator during market downturns with the dot-com bubble burst generator showing the best performance.
Although the averaged generator does not outperform the dot-com bubble burst generator, it still performs well compared to the base generator during market downturns.
As expected, these agents underperform the base generator during the full dataset period since the training is tailored to the market downturns by design.
A portfolio manager could actively switch to a RL agent trained with robust noise during market downturns to improve performance during these periods.
Overall, this method shows promise in improving the robustness of the generative model for training the RL agent to adapt to different market environments without requiring large amounts of data.

\section{Conclusion and Future Work} \label{sec:conclusion}

This paper introduces an approach for generating realistic synthetic financial time series data using MMD with a signature kernel.
The use of a moving average (MA) model to vary the variance of the noise input to the generative model aids the performance.
In particular, the model demonstrates the ability to capture key stylized facts observed in financial markets.

The application of the generative model for training a reinforcement learning agent in the task of portfolio management shows promising results, highlighting the potential of our approach for practical applications.
We also proposed a method to improve the robustness of the generative model by selectively fitting the noise distribution to periods of market downturns.
This method shows promise in capturing the characteristics of market downturns without requiring large amounts of data.

Several avenues exist for further exploration and improvement. The noise structure is clearly important for generating realistic sequences, and further research is needed to better understand the impact of different noise distributions on the generated data.
The robustness of the generative model can be further improved by incorporating additional features or data sources to better capture the complexity of financial markets.
For the RL application, one could introduce a Bayesian approach to learn the prevalent regime of the market and switch between different agents trained on different noise distributions or use an ensemble of weighted agents (see e.g. \cite{duembgen2014estimate}).
{
In our experiments, we found that RL tends to find an averaged strategy when mixing datasets from very different market environments.
It requires additional architectural adjustments to produce distinct strategies for different environments (see \cite{lu2023evaluation}).
Further research is needed to see how such adjustments can be generalised.
}
Expanding empirical studies and applications by handling multiple assets and investigating more sophisticated sampling and data preprocessing techniques are also possible directions for future research.

\section*{Acknowledgments}
\noindent
J. Sester is grateful for financial support
 by the NUS Start-Up Grant \emph{Tackling model uncertainty in Finance with machine learning}.

\bibliographystyle{apalike}
\bibliography{MMD_SGM}

\begin{thebibliography}{}

\bibitem[Acciaio et~al., 2024]{acciaio2024time}
Acciaio, B., Eckstein, S., and Hou, S. (2024).
\newblock Time-causal {VAE}: Robust financial time series generator.
\newblock {\em arXiv preprint arXiv:2411.02947}.

\bibitem[Arjovsky et~al., 2017]{arjovsky2017wasserstein}
Arjovsky, M., Chintala, S., and Bottou, L. (2017).
\newblock Wasserstein generative adversarial networks.
\newblock In {\em Proceedings of the 34th International Conference on Machine
  Learning}, volume~70, pages 214--223. PMLR.

\bibitem[Assefa et~al., 2020]{assefa2020generating}
Assefa, S.~A., Dervovic, D., Mahfouz, M., Tillman, R.~E., Reddy, P., and
  Veloso, M. (2020).
\newblock Generating synthetic data in finance: opportunities, challenges and
  pitfalls.
\newblock In {\em Proceedings of the First ACM International Conference on AI
  in Finance}, pages 1--8.

\bibitem[Backhoff et~al., 2022]{backhoff2022estimating}
Backhoff, J., Bartl, D., Beiglböck, M., and Wiesel, J. (2022).
\newblock Estimating processes in adapted {W}asserstein distance.
\newblock {\em The Annals of Applied Probability}, 32(1):529--550, 22.

\bibitem[Biagini et~al., 2024]{biagini2024universal}
Biagini, F., Gonon, L., and Walter, N. (2024).
\newblock Universal randomised signatures for generative time series modelling.
\newblock {\em arXiv preprint arXiv:2406.10214}.

\bibitem[Bińkowski et~al., 2018]{binkowski2018demystifying}
Bińkowski, M., Sutherland, D.~J., Arbel, M., and Gretton, A. (2018).
\newblock Demystifying {MMD GAN}s.
\newblock In {\em International Conference on Learning Representations}.

\bibitem[Buehler et~al., 2020]{buehler2020data}
Buehler, H., Horvath, B., Lyons, T., Arribas, I.~P., and Wood, B. (2020).
\newblock A data-driven market simulator for small data environments.
\newblock {\em arXiv preprint arXiv:2006.14498}.

\bibitem[Chevyrev and Kormilitzin, 2016]{chevyrev2016primer}
Chevyrev, I. and Kormilitzin, A. (2016).
\newblock A primer on the signature method in machine learning.
\newblock {\em arXiv preprint arXiv:1603.03788}.

\bibitem[Chevyrev and Oberhauser, 2022]{chevyrev2022signature}
Chevyrev, I. and Oberhauser, H. (2022).
\newblock Signature moments to characterize laws of stochastic processes.
\newblock {\em Journal of Machine Learning Research}, 23(176):1--42.

\bibitem[Coletta et~al., 2021]{coletta2021towards}
Coletta, A., Prata, M., Conti, M., Mercanti, E., Bartolini, N., Moulin, A.,
  Vyetrenko, S., and Balch, T. (2021).
\newblock Towards realistic market simulations: a generative adversarial
  networks approach.
\newblock {\em arXiv preprint arXiv:2110.13287}.

\bibitem[Compagnoni et~al., 2023]{compagnoni2023effectiveness}
Compagnoni, E.~M., Scampicchio, A., Biggio, L., Orvieto, A., Hofmann, T., and
  Teichmann, J. (2023).
\newblock On the effectiveness of randomized signatures as reservoir for
  learning rough dynamics.
\newblock In {\em 2023 International Joint Conference on Neural Networks
  (IJCNN)}, pages 1--8. IEEE.

\bibitem[Cont, 2001]{cont2001empirical}
Cont, R. (2001).
\newblock Empirical properties of asset returns: stylized facts and statistical
  issues.
\newblock {\em Quantitative Finance}, 1(2):223--236.

\bibitem[Cont et~al., 2022]{cont2022tail}
Cont, R., Cucuringu, M., Xu, R., and Zhang, C. (2022).
\newblock Tail-gan: Learning to simulate tail risk scenarios.
\newblock {\em arXiv preprint arXiv:2203.01664}.

\bibitem[Duembgen and Rogers, 2014]{duembgen2014estimate}
Duembgen, M. and Rogers, L. C.~G. (2014).
\newblock Estimate nothing.
\newblock {\em Quantitative Finance}, 14(12):2065--2072.

\bibitem[Fermanian, 2021]{fermanian2021embedding}
Fermanian, A. (2021).
\newblock Embedding and learning with signatures.
\newblock {\em Computational Statistics and Data Analysis}, 157:107148.

\bibitem[Freund et~al., 1997]{freund1997market}
Freund, W.~C., Larrain, M., and Pagano, M.~S. (1997).
\newblock Market efficiency before and after the introduction of electronic
  trading at the toronto stock exchange.
\newblock {\em Review of Financial Economics}, 6(1):29--56.

\bibitem[Fu et~al., 2022]{fu2022simulating}
Fu, W., Hirsa, A., and Osterrieder, J. (2022).
\newblock Simulating financial time series using attention.
\newblock {\em arXiv preprint arXiv:2207.00493}.

\bibitem[Goerg, 2015]{goerg2015lambert}
Goerg, G.~M. (2015).
\newblock The {Lambert} way to gaussianize heavy-tailed data with the inverse
  of {Tukey’s} h transformation as a special case.
\newblock {\em The Scientific World Journal}, 2015.

\bibitem[Goodfellow et~al., 2016]{goodfellow2016deep}
Goodfellow, I., Bengio, Y., and Courville, A. (2016).
\newblock {\em Deep Learning}.
\newblock MIT press.

\bibitem[Goodfellow et~al., 2014]{goodfellow2014generative}
Goodfellow, I., Pouget-Abadie, J., Mirza, M., Xu, B., Warde-Farley, D., Ozair,
  S., Courville, A., and Bengio, Y. (2014).
\newblock Generative adversarial nets.
\newblock In {\em Advances in Neural Information Processing Systems},
  volume~27.

\bibitem[Gretton et~al., 2012]{gretton2012kernel}
Gretton, A., Borgwardt, K.~M., Rasch, M.~J., Schölkopf, B., and Smola, A.
  (2012).
\newblock A kernel two-sample test.
\newblock {\em Journal of Machine Learning Research}, 13:723–773.

\bibitem[Hambly and Lyons, 2010]{hambly2010uniqueness}
Hambly, B. and Lyons, T. (2010).
\newblock Uniqueness for the signature of a path of bounded variation and the
  reduced path group.
\newblock {\em Annals of Mathematics}, 171(1):109--167.

\bibitem[Heston, 1993]{heston1993closed}
Heston, S.~L. (1993).
\newblock A closed-form solution for options with stochastic volatility with
  applications to bond and currency options.
\newblock {\em The review of financial studies}, 6(2):327--343.

\bibitem[Ho et~al., 2020]{ho2020denoising}
Ho, J., Jain, A., and Abbeel, P. (2020).
\newblock Denoising diffusion probabilistic models.
\newblock {\em Advances in Neural Information Processing Systems},
  33:6840--6851.

\bibitem[Hochreiter and Schmidhuber, 1997]{hochreiter1997long}
Hochreiter, S. and Schmidhuber, J. (1997).
\newblock Long short-term memory.
\newblock {\em Neural computation}, 9(8):1735--1780.

\bibitem[Huang et~al., 2024]{huang2024generative}
Huang, H., Chen, M., and Qiao, X. (2024).
\newblock Generative learning for financial time series with irregular and
  scale-invariant patterns.
\newblock In {\em The Twelfth International Conference on Learning
  Representations}.

\bibitem[Issa et~al., 2023]{issa2023non}
Issa, Z., Horvath, B., Lemercier, M., and Salvi, C. (2023).
\newblock Non-adversarial training of neural {SDE}s with signature kernel
  scores.
\newblock {\em arXiv preprint arXiv:2305.16274}.

\bibitem[Kidger, 2022]{kidger2022neural}
Kidger, P. (2022).
\newblock On neural differential equations.
\newblock {\em arXiv preprint arXiv:2202.02435}.

\bibitem[Kidger et~al., 2021]{kidger2021neural}
Kidger, P., Foster, J., Li, X., and Lyons, T.~J. (2021).
\newblock Neural {SDE}s as infinite-dimensional {GAN}s.
\newblock In {\em Proceedings of the 38th International Conference on Machine
  Learning}, volume 139, pages 5453--5463. PMLR.

\bibitem[Kingma and Welling, 2013]{kingma2013auto}
Kingma, D.~P. and Welling, M. (2013).
\newblock Auto-encoding variational {Bayes}.
\newblock {\em arXiv preprint arXiv:1312.6114}.

\bibitem[Király and Oberhauser, 2019]{kiraly2019kernels}
Király, F.~J. and Oberhauser, H. (2019).
\newblock Kernels for sequentially ordered data.
\newblock {\em Journal of Machine Learning Research}, 20.

\bibitem[Kondratyev and Schwarz, 2019]{kondratyev2019market}
Kondratyev, O. and Schwarz, C. (2019).
\newblock The market generator.
\newblock {\em Econometrics: Econometric \& Statistical Methods - Special
  Topics eJournal}.

\bibitem[Kong et~al., 2020]{kong2020diffwave}
Kong, Z., Ping, W., Huang, J., Zhao, K., and Catanzaro, B. (2020).
\newblock Diffwave: A versatile diffusion model for audio synthesis.
\newblock {\em arXiv preprint arXiv:2009.09761}.

\bibitem[Koshiyama et~al., 2021]{koshiyama2021generative}
Koshiyama, A., Firoozye, N., and Treleaven, P. (2021).
\newblock Generative adversarial networks for financial trading strategies
  fine-tuning and combination.
\newblock {\em Quantitative Finance}, 21(5):797--813.

\bibitem[LeCun et~al., 2002]{lecun2002efficient}
LeCun, Y., Bottou, L., Orr, G.~B., and Müller, K.-R. (2002).
\newblock {\em Efficient backprop}, pages 9--50.
\newblock Springer.

\bibitem[Liao et~al., 2024]{liao2024sig}
Liao, S., Ni, H., Sabate-Vidales, M., Szpruch, L., Wiese, M., and Xiao, B.
  (2024).
\newblock {Sig-Wasserstein GANs} for conditional time series generation.
\newblock {\em Mathematical Finance}, 34(2):622--670.

\bibitem[Lozano et~al., 2023]{lozano2023neural}
Lozano, P.~D., Bagén, T.~L., and Vives, J. (2023).
\newblock Neural stochastic differential equations for conditional time series
  generation using the {Signature-Wasserstein-1} metric.
\newblock {\em Journal of Computational Finance}, 27(1).

\bibitem[Lu, 2023]{lu2023evaluation}
Lu, C.~I. (2023).
\newblock Evaluation of deep reinforcement learning algorithms for portfolio
  optimisation.
\newblock {\em arXiv preprint arXiv:2307.07694}.

\bibitem[Lyons, 2014]{lyons2014rough}
Lyons, T. (2014).
\newblock Rough paths, signatures and the modelling of functions on streams.
\newblock {\em arXiv preprint arXiv:1405.4537}.

\bibitem[Lyons and McLeod, 2022]{lyons2022signature}
Lyons, T. and McLeod, A.~D. (2022).
\newblock Signature methods in machine learning.
\newblock {\em arXiv preprint arXiv:2206.14674}.

\bibitem[Lyons et~al., 2007]{lyons2007differential}
Lyons, T.~J., Caruana, M., and Lévy, T. (2007).
\newblock {\em Differential equations driven by rough paths}.
\newblock Springer.

\bibitem[Lyons and Xu, 2018]{lyons2018inverting}
Lyons, T.~J. and Xu, W. (2018).
\newblock Inverting the signature of a path.
\newblock {\em Journal of the European Mathematical Society}, 20(7):1655--1687.

\bibitem[Muandet et~al., 2016]{muandet2016kernel}
Muandet, K., Fukumizu, K., Sriperumbudur, B., and Schölkopf, B. (2016).
\newblock Kernel mean embedding of distributions: A review and beyond.
\newblock {\em arXiv preprint arXiv:1605.09522}.

\bibitem[Ni et~al., 2021]{ni2021sig}
Ni, H., Szpruch, L., Sabate-Vidales, M., Xiao, B., Wiese, M., and Liao, S.
  (2021).
\newblock {Sig-Wasserstein GANs} for time series generation.
\newblock In {\em Proceedings of the Second ACM International Conference on AI
  in Finance}, pages 1--8.

\bibitem[Rame et~al., 2022]{rame2022diverse}
Rame, A., Kirchmeyer, M., Rahier, T., Rakotomamonjy, A., Gallinari, P., and
  Cord, M. (2022).
\newblock Diverse weight averaging for out-of-distribution generalization.
\newblock {\em Advances in Neural Information Processing Systems},
  35:10821--10836.

\bibitem[Salvi et~al., 2021]{salvi2021signature}
Salvi, C., Cass, T., Foster, J., Lyons, T., and Yang, W. (2021).
\newblock The signature kernel is the solution of a {Goursat PDE}.
\newblock {\em SIAM Journal on Mathematics of Data Science}, 3(3):873--899.

\bibitem[Schulman et~al., 2017]{schulman2017proximal}
Schulman, J., Wolski, F., Dhariwal, P., Radford, A., and Klimov, O. (2017).
\newblock Proximal policy optimization algorithms.
\newblock {\em arXiv preprint arXiv:1707.06347}.

\bibitem[Sheppard, 2024]{sheppard2024bashtage/arch}
Sheppard, K. (2024).
\newblock bashtage/arch: Release 6.3.

\bibitem[Song and Ermon, 2019]{song2019generative}
Song, Y. and Ermon, S. (2019).
\newblock Generative modeling by estimating gradients of the data distribution.
\newblock In {\em Advances in Neural Information Processing Systems},
  volume~32.

\bibitem[Sriperumbudur et~al., 2011]{sriperumbudur2011universality}
Sriperumbudur, B.~K., Fukumizu, K., and Lanckriet, G. R.~G. (2011).
\newblock Universality, characteristic kernels and {RKHS} embedding of
  measures.
\newblock {\em J. Mach. Learn. Res.}, 12:2389–2410.

\bibitem[Steinwart and Christmann, 2008]{steinwart2008support}
Steinwart, I. and Christmann, A. (2008).
\newblock {\em Support vector machines}.
\newblock Springer Science \& Business Media.

\bibitem[Sutherland and Deka, 2019]{sutherland2019unbiased}
Sutherland, D.~J. and Deka, N. (2019).
\newblock Unbiased estimators for the variance of {MMD} estimators.
\newblock {\em arXiv preprint arXiv:1906.02104}.

\bibitem[Takahashi et~al., 2019]{takahashi2019modeling}
Takahashi, S., Chen, Y., and Tanaka-Ishii, K. (2019).
\newblock Modeling financial time-series with generative adversarial networks.
\newblock {\em Physica A: Statistical Mechanics and its Applications},
  527:121261.

\bibitem[Wiese et~al., 2020]{wiese2020quant}
Wiese, M., Knobloch, R., Korn, R., and Kretschmer, P. (2020).
\newblock Quant {GAN}s: deep generation of financial time series.
\newblock {\em Quantitative Finance}, 20(9):1419--1440.

\bibitem[Wiese et~al., 2021]{wiese2021multi}
Wiese, M., Wood, B., Pachoud, A., Korn, R., Buehler, H., Murray, P., and Bai,
  L. (2021).
\newblock Multi-asset spot and option market simulation.
\newblock {\em arXiv preprint arXiv:2112.06823}.

\bibitem[Williams and Rasmussen, 2006]{williams2006gaussian}
Williams, C.~K. and Rasmussen, C.~E. (2006).
\newblock {\em Gaussian processes for machine learning}, volume~2.
\newblock MIT press Cambridge, MA.

\bibitem[Wortsman et~al., 2022]{wortsman2022model}
Wortsman, M., Ilharco, G., Gadre, S.~Y., Roelofs, R., Gontijo-Lopes, R.,
  Morcos, A.~S., Namkoong, H., Farhadi, A., Carmon, Y., Kornblith, S., and
  Schmidt, L. (2022).
\newblock Model soups: averaging weights of multiple fine-tuned models improves
  accuracy without increasing inference time.
\newblock In {\em Proceedings of the 39th International Conference on Machine
  Learning}, volume 162, pages 23965--23998. PMLR.

\bibitem[Xu et~al., 2020]{xu2020cot}
Xu, T., Wenliang, L.~K., Munn, M., and Acciaio, B. (2020).
\newblock {COT-GAN}: Generating sequential data via causal optimal transport.
\newblock In {\em Advances in Neural Information Processing Systems},
  volume~33, pages 8798--8809.

\bibitem[Yoon et~al., 2019]{yoon2019time}
Yoon, J., Jarrett, D., and van~der Schaar, M. (2019).
\newblock Time-series generative adversarial networks.
\newblock In {\em Advances in neural information processing systems},
  volume~32.

\end{thebibliography}

\section*{Appendix}
\appendix
\section{Heston Model Experiment} \label{app:heston}
The Heston model is defined by the following stochastic differential equations:
\begin{align*}
    dS_t &= \mu S_t dt + \sqrt{v_t} S_t dW_t^1, \\
    dv_t &= \kappa (\theta - v_t) dt + \sigma \sqrt{v_t} dW_t^2,
\end{align*}
where $S_t$ is the asset price, $v_t$ is the variance of the asset price, $\mu$ is the drift, $\kappa$ is the mean reversion rate, $\theta$ is the long-term variance, $\sigma$ is the volatility of the variance, $W_t^1$ and $W_t^2$ are Brownian motions where $dW_t^1 dW_t^2 = \rho dt$ for some correlation parameter $\rho \in [-1,1]$.

We set the parameters $\mu = 0.2$, $\kappa = 1$, $\theta = 0.25$, $\sigma = 0.7$, $\rho = -0.7$ and the initial variance $v_0 = 0.09$.
Using the Quantlib python library \url{https://quantlib-python-docs.readthedocs.io/en/latest/}, we generate 6,400 samples of length 250 with a time step of 1/252.
The Quantlib implementation uses first order Euler discretisation to simulate trajectories of the Heston model.

The models are trained using a simplified version of Algorithm \ref{alg:training}, where the conditioning process is omitted the LSTM cell state and hidden state and only the time augmentation is used instead of both lead-lag and time augmentation.

To simulate noise, we generate independent paths using the same parameters for the Heston process, except with $\mu=0$, then calculate $z_{t_i} = \frac{\log{S_{t_i}} - \log{S_{t_{i-1}}}}{\sqrt{(t_i-t_{i-1}) v_0}}$ to obtain the noise at time $t_i$.
For the i.i.d.\,standard Gaussian noise, we use the same Quantlib Heston implementation but set $\kappa = \sigma = 1e-9$ and $\theta = v_0$ which effectively results in the noise having a mean of zero and variance of one.
This method of simulating the i.i.d.\,Gaussian noise allows the use of the same random seed and generate noise sequences that are only different in variance so as to provides a fair comparison with the noise that has changing variance.
We use a noise dimension of 2 as the Heston model has two Brownian motions.

The truncation level was set to $m=5$ and the static kernel was a rational quadratic kernel with $\alpha=1$ and $l=0.1$.

\section{Moments of Returns} \label{app:returns_stats}

We briefly present the formulas used to calculate the empirical estimates of the moments of the returns in Section \ref{sec:stylised_facts}.
Let $(r_1, r_2, \ldots, r_n)$ be a sample of log returns we define $\bar{r} := \frac{1}{n} \sum_{i=1}^{n} r_i$ and $m_i := \frac{1}{n} \sum_{j=1}^{n} (r_j - \bar{r})^i$.
We will assume that there are 252 business days in a year.

The first four moments of the log returns are named and calculated as follows:
\begin{enumerate}
    \item The mean is annualised by business day convention which we calculate it as: $\textbf{Ann. returns} = 252 \times \bar{r}$.
    \item Similarly volatility is also annualised by convention and calculated as: $\textbf{Volatility} = \sqrt{252 \times m_2}$.
    \item The skew is calculated as: $\textbf{Skew} = \frac{m_3}{m_2^{3/2}}$.
    \item The kurtosis is calculated as: $\textbf{Kurtosis} = \frac{m_4}{m_2^2} - 3$.
\end{enumerate}

\end{document}